\documentclass[10pt]{article}

\usepackage[T1]{fontenc}
\usepackage{lmodern}
\usepackage[frenchmath]{mathastext}	
\usepackage{microtype}											
\microtypesetup{expansion=false}						
\usepackage[authoryear]{natbib}       									
\usepackage{amsmath}                        
\numberwithin{equation}{section}
\usepackage{enumitem}	   
\usepackage{mdwlist}		
\usepackage[dvipsnames]{xcolor}	
\usepackage[plainpages=false, pdfpagelabels]{hyperref} 
	\hypersetup{
		colorlinks   = true,
		citecolor    = RoyalBlue, 
		linkcolor    = RubineRed, 
		urlcolor     = Turquoise
	}
\usepackage{amssymb}                        
\usepackage[mathscr]{eucal}                 
\usepackage{dsfont}													
\usepackage[paperwidth=8.5in,paperheight=11in,top=1.25in, bottom=1.25in, left=1.0in, right=1.0in]{geometry} 
\usepackage{mathtools}                      
	\mathtoolsset{showonlyrefs=true}          
\usepackage{amsthm}                         

\usepackage{graphicx}
\usepackage{epstopdf}
\usepackage{empheq}

	\allowdisplaybreaks                       
	\theoremstyle{plain}
	\newtheorem{theorem}{Theorem}
	\numberwithin{theorem}{section}
	\newtheorem{lemma}[theorem]{Lemma}       	
	\newtheorem{proposition}[theorem]{Proposition}
	
	\theoremstyle{definition}

	\newtheorem{remark}[theorem]{Remark}
	\newtheorem{assumption}[theorem]{Assumption}
    \newtheorem{problem}[theorem]{Problem}

\usepackage{bbm}
\newcommand{\F}{\mathscr{F}}

\newcommand{\mb}{\mathbb{M}}
\newcommand{\vb}{\mathbbm{v}}

\renewcommand{\(}{\left(}
\renewcommand{\)}{\right)}
\renewcommand{\[}{\left[}
\renewcommand{\]}{\right]}

\newcommand\Eb{\mathds{E}}
\newcommand\Fb{\mathds{F}}
\newcommand\Ib{\mathds{1}}

\newcommand\Pb{\mathds{P}}

\newcommand\Rb{\mathds{R}}

\newcommand\Vb{\mathds{V}}

\newcommand\Ac{\mathscr{A}}

\newcommand\Cc{\mathscr{C}}

\newcommand\Nc{\mathscr{N}}

\newcommand\eps{\varepsilon}

\newcommand\Om{\Omega}
\newcommand\sig{\sigma}

\newcommand\lam{\lambda}

\newcommand\Del{\Delta}

\newcommand\kf{{\mathfrak{f}}}

\newcommand\kg{{\mathfrak{g}}}

\newcommand\dd{\mathrm{d}}
\newcommand\ee{\mathrm{e}}

\begin{document}

\title{
A Mathematical Analysis of Technical Analysis}

\author{
Matthew Lorig\thanks{Department of Applied Mathematics, University of Washington, Seattle, WA, USA. Email: mlorig@uw.edu}, \quad 
Zhou Zhou\thanks{School of Mathematics and Statistics, University of Sydney, Sydney, Australia. Email: zhou.zhou@sydney.edu.au}, \quad 
Bin Zou\thanks{Department of Mathematics, University of Connecticut, Storrs, CT, USA. Email: bin.zou@uconn.edu}
}

\date{This version: 
\today
}

\maketitle

\begin{abstract}
In this paper, we investigate trading strategies based on exponential moving averages (ExpMAs) of an underlying risky asset.
We study both logarithmic utility maximization and long-term growth rate maximization problems and find closed-form solutions when the drift of the underlying is modeled by either an Ornstein-Uhlenbeck process or a two-state continuous-time Markov chain.
For the case of an Ornstein-Uhlenbeck drift, we carry out several Monte Carlo experiments in order to investigate how the performance of optimal ExpMA strategies is affected by variations in model parameters and by transaction costs.  
\end{abstract}

\\[0.5em]

%
%

\section{Introduction}
\label{sec_introduction}

Technical analysis is a methodology for forecasting the future movements of securities prices by analyzing past market data (most often, but not limited to, prices and trading volumes).  
Within technical analysis, there are many indicators, which purport to provide information about the future direction and volatility of an underlyer (e.g., a stock, currency, interest rate, etc.).  These indicators often have gimmicky names such as, e.g., Smart Money Index, Know Sure Thing Oscillator, Vortex Indicator, Money Flow Index, Bollinger Bands, etc..  From a mathematical standpoint, perhaps the simplest indicator to construct (and one which has an uncharacteristically boring name) is the Moving Average.  As the name suggests, a moving average $Y=(Y_t)_{t \ge 0}$ of a process $X = (X_t)_{t \ge 0}$ is constructed via a convolution of $X$ with a kernel $\rho$.  Specifically,
\begin{align}
Y_t
&=	\int_{0}^t \rho(t-s) \cdot X_s \dd s , &
&\text{where}&
\rho
&\geq 0 , &
&\text{and}&
\int_0^\infty \rho(t) \dd t
&= 1 . 
\end{align}
Common moving averages are the \emph{Simple moving average} (SimMA): $\rho(t) = \Ib_{[0,T]}(t)/T$, where $\Ib$ is an indicator function, and the \emph{Exponential moving average} (ExpMA): $\rho(t) = \lam \ee^{- \lam t}$, where $\lam > 0$.

The main purpose of this paper is to provide a mathematical analysis of trading strategies based on ExpMAs, which are \emph{observable} technical indicators.
To be clear, our aim is \emph{not} to support or disapprove the use of technical analysis in portfolio management. 
We consider utility maximization and long-term growth rate maximization problems for trading strategies based on ExpMAs. Optimal ExpMA strategies are obtained (semi-) explicitly when the drift process of the risky asset is modeled by an Ornstein-Uhlenbeck (OU) process or a continuous-time Markov chain (CTMC), which is often used in a regime switching market. 
In numerical studies, we carry out Monte Carlo simulations to test the performance of optimal ExpMA strategies against a buy-and-hold strategy. 
In addition, we conduct a sensitivity analysis and investigate the impact of transaction costs on optimal ExpMA strategies. 
In general, our simulation results show that optimal ExpMA strategies deliver excellent returns.
However, if an investor's measure of portfolio performance is the Sharpe ratio, or if transaction costs are present, ExpMA-based strategies may not be optimal.
To the best of our knowledge, there is no paper studying optimal ExpMA-based
trading strategies in literature. Our paper will fill this void
and also provide valuable guidance for practitioners who use moving averages to trade securities.


In mathematical finance, optimal investment problems are well studied, c.f., the classical works of \cite{markowitz1952portfolio} and \cite{merton1969lifetime, merton1971optimum}. 
\cite{kim1996dynamic} and \cite{wachter2002portfolio} provide analytical solutions to utility maximization problems (similar to Problem \ref{prob_utility}) when the drift of the stock price is given by an \emph{observable} OU process. 
\cite{bauerle2004portfolio} solve the same problems when the drift of the stock price is modeled by an \emph{observable} CTMC. 
Long-term growth rate maximization problems (similar to Problem \ref{prob_growth}) have been studied by \cite{fleming1995risk} and \cite{fleming1999optimal}. 
Notice that all papers mentioned above assume investors can observe the drift process at all times. 
Our paper is related to utility maximization problems under partial information, in which the drift process is  \emph{unobservable} to investors in the market. 
\cite{lakner1995utility, lakner1998optimal} consider such problems for an \emph{unobservable} OU drift while
\cite{honda2003optimal} and \cite{sass2004optimizing} consider such problems for an \emph{unobservable} CTMC drift. 
The ExpMA trading strategies considered in this paper fall under the category of trend following strategies. 
\cite{dai2010trend, dai2016optimal} consider optimal stopping times problems in an \emph{unobservable} regime switching market (two-state Markov chain drift) for an investor who chooses a sequence of buying and selling times to  maximize the net gain. 
{They show that the optimal trading strategy is trend following, and is superior to a buy-and-hold strategy.}

The standard method of dealing with an unobservable drift is to apply the Wonham filter (see \cite{wonham1964some}) to transform the problem with partial information to the one with full information (via the innovation process), which is used in \cite{lakner1995utility, lakner1998optimal}, \cite{honda2003optimal} \cite{sass2004optimizing}, and \cite{dai2010trend, dai2016optimal}. 
Once full information is gained, one may apply either the martingale method (see \cite{lakner1995utility} and \cite{sass2004optimizing}) or the HJB method (see \cite{honda2003optimal} and \cite{dai2010trend}) to obtain optimal solutions. Using ExpMAs to find optimal trading strategies is a completely different methodology because, although the drift process is assumed to be unobservable, the \emph{observable} ExpMA of the risky asset is used to deduce information about drift and to construct trading strategies.

The rest of this paper proceeds as follows.
In Section \ref{sec_setup}, we introduce a market model in which a risky asset has a stochastic drift.
We also describe the optimal investment problems we wish to consider for ExpMA strategies and present a general solution to one of the problems.
In Section \ref{sec_OU}, we obtain optimal ExpMA strategies in explicit form when the drift is modeled as an OU process.
In Section \ref{sec_MC}, we obtain optimal ExpMA strategies when the drift is modeled as a two-state CTMC. 
Numerical studies are presented in Section \ref{sec:monte.carlo}.
Some concluding remarks are offered in Section \ref{sec_conclusion}.
Technical proofs are given in Appendix \ref{sec_appendix}.

\section{Modeling Framework and General Solutions}
\label{sec_setup}

\subsection{The Model}

We now turn our attention to the mathematical analysis of moving average strategies.
We consider a continuous-time financial market, which consists of one riskless asset and one risky asset.
For simplicity, we assume that the risk-free rate of interest is zero so that the riskless asset has a constant value.
The price process of the risky asset $S = (S_t)_{t \geq 0}$
is given by the following dynamics under some given stochastic basis
$(\Om,\F,\Fb= (\F_t)_{t \geq 0},\Pb)$
\begin{equation}\label{eq:dS}
\dd S_t = \mu_t S_t \dd t + \sigma S_t \dd W_t,\quad t\ge 0,\quad\text{and}\quad S_0>0, 
\end{equation}
where the drift $\mu=(\mu_t)_{t \geq 0}$ is $\Fb$-adapted, 
the volatility $\sigma$ is a positive constant,
 and $W$ is a standard one-dimensional Brownian Motion under $\Pb$ with respect to the filtration $\Fb$.
Throughout this paper, we shall assume that $\mu$ 
conspires so that the solution $S$ of \eqref{eq:dS} exists and is strictly positive for all $t \geq 0$.

In a classical portfolio optimization problem, one seeks to solve
\begin{align}
\sup_{\pi \in \Ac} \; \Eb  \big[ U(\Pi_T^\pi) \big], \label{eq:merton.problem}
\end{align}
where $T > 0$ is the terminal time (or planning time), 
$U$ is some utility function, $\pi=(\pi_t)_{t \geq 0}$ is the investor's strategy with $\pi_t$ denoting the investment proportion in the risky asset at time $t$, $\Ac$ is some set of admissible strategies, and $\Pi^\pi = (\Pi^\pi_t)_{t \in [0,T]}$ is the wealth process associated with strategy $\pi$, with dynamics given by
\begin{equation}
\dd \Pi_t^\pi=\frac{\pi_t \, \Pi_t^\pi}{S_t} \dd S_t,\ \ t\in [0,T],\quad\text{and}\quad\Pi_0> 0. 
\end{equation}
In general, the optimal strategy $\pi^*$ depends on knowing  the drift value $\mu_t$ at all times $t \in [0,T]$. 
For instance, in the classical Merton's problem under logarithmic utility, $\pi^*_t = \frac{\mu_t}{\sigma^2}$ for all $t\in [0,T]$.
However, the instantaneous value $\mu_t$ of the drift is often \emph{unobservable}. One way of dealing with this, is to use filtering to estimate $\mu_t$ and derive the optimal strategy based on one's best estimate of $\mu_t$, denoted by $\widehat{\mu}_t$.
In our studies, we use exponential moving averages to deduce information about the drift.

Let us introduce the $\log$ stock price process $X=(X_t)_{t \geq 0}$, which is defined as $X_t := \ln S_t$.  Using It\^o's Lemma, the dynamics of $X$ are given by
\begin{equation}
\label{eqn_X}
\dd X_t=\big( \mu_t - \tfrac{1}{2} \sig^2 \big) \dd t + \sig \dd W_t ,\quad\text{with}\quad
X_0=\ln S_0 .
\end{equation}
Next, we define $Y=(Y_t)_{t \geq 0}$, the \emph{exponential moving average} (ExpMA) of $X$, by
\begin{equation}
\label{eqn_def_ExpMA}
Y_t:=\int_{0}^t \lam \ee^{-\lam(t-s)} X_s \dd s ,\quad t\in [0,T] ,
\end{equation}
where $\lambda>0$ is a constant.
One can easily check that
\begin{align}
\label{eqn_dY}
\dd Y_t
	&=	\lam (X_t - Y_t) \dd t .
\end{align}
One advantage of ExpMAs over all other MAs is that ExpMAs are Markovian, as seen in the dynamics above.

Note that $Y$ mean-reverts to $X$.  If the drift of $X$ is positive, then, at a given time $t$, we will likely have that $Y_t$ is less than $X_t$.  The larger the drift of $X$ is, the larger
the gap between $X_t$ and $Y_t$ will be.  Thus, the quantity $X_t-Y_t$ can provide information about the drift of $X$.  
Note that
$X_t - Y_t$ is easily \emph{observable}.
This motivates us to consider trading strategies of the form
\begin{equation}\label{eq:form}
\pi_t=f(t, X_t - Y_t),\quad t\in [0,T] , 
\end{equation}
where $f$ is some increasing function of the second argument. 
The positive monotonicity of $f$ implies that the ExpMA strategies considered in this paper fall under the category of trend following trading rules.

It will be useful at this point to define  the difference process $Z = (Z_t)_{t \geq 0}$, which is given by $Z_t := X_t - Y_t$.  One can easily verify that the dynamics of $Z$ are given by
\begin{align}
\dd Z_t
	&=	\lam \Big( \frac{\mu_t - \tfrac{1}{2} \sig^2}{\lam} - Z_t \Big) \dd t + \sig \dd W_t .
\end{align}
Solving the stochastic differential equation (SDE) for $Z$, we obtain
\begin{align}
\label{eqn_Z_sol}
Z_t &= Z_0 + e^{-\lambda t} \left( \int_0^t e^{\lambda s} \left(\mu_s - \frac{1}{2} \sig^2 \right) \dd s + \sig \int_0^t e^{\lambda s} \dd W_s \right).
\end{align}

Note that if $\mu$ is a constant, $Z$ is simply an OU process.  Note further that $Z$ is \emph{not} independent of $X$ as both processes are driven by the same Brownian motion $W$.  Using our definition of $Z$, we can write strategies of the form \eqref{eq:form} as follows
\begin{equation}
\pi_t:=f(t, Z_t),\quad t \in [0,T].
\end{equation}
We will denote by $\Pi^f = (\Pi_t^f)_{t \in [0,T]}$ the wealth processes corresponding to investment strategy $\pi = (\pi_t)_{t \in [0,T]} = (f(t,Z_t))_{t \in [0,T]}$. Note $\Pi_0$ is the same for all strategies $f$.

We are primarily interested in the following portfolio optimization problems for the ExpMA strategies.
\begin{problem}
\label{prob_utility}
Find the optimal strategy $f^* \in \Cc_i$ to the utility optimization problem
\begin{align}
\sup_{f \in \Cc_i} \; \Eb  \left[ \ln \left( \frac{\Pi_T^f}{\Pi_0} \right) \right], \quad i = 1,2 ,
\end{align}
where $T>0$ is the planning time, $\Cc_1$ is the set of \emph{affine} strategies
\begin{align}
\Cc_1
	&:=	 \left\{ f : [0,T] \times \mathbb{R} \to \mathbb{R} \; | \; f(t,z) = a \cdot z + b, \quad \text{where } a, b \in \mathbb{R} \right\}. \label{eqn_def_Cc1}	
\end{align}
and $\Cc_2$ is the set of \emph{square-integrable} strategies
\begin{align}
\Cc_2
	&:=	\left\{f:[0,T]\times\mathbb{R}\mapsto \mathbb{R}\, \bigg| \; \Eb \left[\int_0^Tf^2(t,Z_t)\, \dd t \right]<\infty \right\} .  \label{eqn_def_Cc2}
\end{align}
\end{problem}

\begin{problem}
\label{prob_growth}
Find the optimal strategy $f^* \in \Cc_i$ to maximize the long-term growth rate
\begin{align}
\sup_{f \in \Cc_i} \; \lim_{T \to \infty} \frac{1}{T} \, \Eb  \left[ \ln \left( \frac{\Pi_T^f}{\Pi_0} \right) \right], \quad i = 1,2 .
\end{align}
\end{problem}

\subsection{General Solutions to Problem \ref{prob_utility}}
\label{subsec_general_solutions}

In this section,
we solve Problem \ref{prob_utility} for $\Cc_1$ and $\Cc_2$ Strategies in the most general setting (i.e., no assumption on the dynamics of the drift $\mu$), and present the results in Theorems \ref{theorem_general_utility_C1} and \ref{theorem_general_utility_C2}, respectively.

For any fixed $T>0$, introduce the following notations:
\begin{align}
\label{eqn_ABCD}
A(T) := \int_0^T \Eb [\mu_t Z_t] \dd t, \quad B(T):= \int_0^T \Eb[\mu_t] \dd t, \quad 
C(T):= \int_0^T \Eb[Z_t^2] \dd t, \quad D(T) := \int_0^T \Eb[Z_t] \dd t. \quad
\end{align} 
The theorem below solves Problem \ref{prob_utility} for all affine strategies (i.e., $\Cc_1$ strategies, see definition in \eqref{eqn_def_Cc1}).

\begin{theorem}
	\label{theorem_general_utility_C1}
The optimal ExpMA strategy $f_1^* \in \Cc_1$ to Problem \ref{prob_utility} is 
\begin{align}
\label{eqn_utility_C1_solutions}
f_1^*(t,z) = a_1^* \cdot z + b_1^*,
\end{align}
where $a_1^*$ and $b_1^*$ are given by 
\begin{align}
\label{eqn_utility_C1_a1b1}
\begin{bmatrix}
a_1^*\\
b_1^*
\end{bmatrix}
=   \frac{1}{(C(T)T-D^2(T))\sigma^2}
\begin{bmatrix}
A(T)T-B(T)D(T)\\
B(T)C(T)-A(T)D(T)
\end{bmatrix},
\end{align}
with $A(T)$, $B(T)$, $C(T)$ and $D(T)$ defined in \eqref{eqn_ABCD}.
\end{theorem}

\begin{proof}

By taking expectation of $\ln (\Pi_T^f / \Pi_0)$ and using \eqref{eqn_ABCD}, we obtain
\begin{align}
\Eb \left[\ln \frac{\Pi_T^f}{\Pi_0} \right] 
&= A(T)a+B(T)b-\frac{\sigma^2}{2}\left(C(T)a^2+2D(T)ab+Tb^2\right) := g(a,b;T). \label{eqn_g}
\end{align}

For any fixed $T>0$, it is easy to see that the function $g(\cdot,\cdot;T)$ attains the global maximum 
at $(a_1^*, b_1^*)$, which are given by \eqref{eqn_utility_C1_a1b1}. 
Noticing $(Z_t)_{t\geq 0}$ is not constant, by the Cauchy-Schwarz inequality, we have $C(T)T-D^2(T) > 0$. 
\end{proof}

Next, we consider Problem \ref{prob_utility} for all square-integrable strategies (i.e., $\Cc_2$ strategies, see definition in \eqref{eqn_def_Cc2}) and summarize the results below.

\begin{theorem}
\label{theorem_general_utility_C2}
If the following condition holds,
\begin{align}
\label{eqn_square_integrable_condition}
\int_0^T \left( \dfrac{\Eb [\mu_t | Z_t]}{\sig^2} \right)^2 \dd t < \infty,
\end{align}
then the optimal ExpMA strategy $f_2^* \in \Cc_2$ to Problem \ref{prob_utility} is 
\begin{align}
\label{eqn_utility_C2_solutions}
f_2^*(t, Z_t) = \dfrac{\Eb [\mu_t | Z_t]}{\sig^2}.
\end{align}
\end{theorem}

\begin{proof}
Given $f \in \Cc_2$, from the SDE of $\Pi^f$, we obtain
\begin{align}
\Eb \left[\ln \frac{\Pi_T^f}{\Pi_0}\right]&=\Eb \left[\int_0^T \left( \mu_t \cdot f(t,Z_t) - \frac{1}{2}\sigma^2 \cdot f^2(t,Z_t) \right) \dd t\right]\\
&=\int_0^T \Eb\left[\left(\Eb\left[\mu_t|Z_t\right] \cdot f(t,Z_t) - \frac{1}{2}\sigma^2 \cdot f^2(t,Z_t) \right)\right] \dd t.
\end{align}
The desired result is then obvious.
\end{proof}

\begin{remark}
The general solution to Problem \ref{prob_growth} is not available for either $\Cc_1$ or $\Cc_2$ strategies. 
In the case of an OU-type drift, the optimal strategy to Problem \ref{prob_growth} is the same for both $\Cc_1$ and $\Cc_2$ strategies, see Theorem \ref{thm_OU_growth}. 
In comparison, when the drift is modeled by a two-state Markov chain, the optimal strategy to Problem \ref{prob_growth} is dramatically different for $\Cc_1$ and $\Cc_2$ strategies, see Theorems \ref{thm_MC_growth_C1} and \ref{thm_MC_growth_C2}. 
\end{remark}

\section{Analysis for the Case of an OU-Type Drift}
\label{sec_OU}

In this section, we study Problems \ref{prob_utility} and \ref{prob_growth} when the drift $\mu$ is given by an OU process.
The main results of this section are Theorems \ref{thm_OU_utility_linear}, \ref{thm_OU_utility_general} and \ref{thm_OU_growth}, where we present the solutions to Problems \ref{prob_utility} and \ref{prob_growth}. 
We make the following two assumptions for the analysis in this section.

\begin{assumption}
	\label{assumption_OU}
The drift $\mu$ follows an OU process,
\begin{align}\label{eqn_OU_drift}
\dd \mu_t
		&=  \kappa \left(\bar{\mu} - \mu_t  \right) \dd t + \delta \dd \bar{W}_t, 
&& t \in [0,T],
\end{align}
where $\kappa$ and $\delta$ are positive constants, $\bar{\mu}$ is the mean-reversion parameter, 
and $\bar{W}$ is a standard Brownian motion, independent of $W$.
We assume
$\kappa \neq \lambda$, where $\lambda$ is the exponential moving average constant, see \eqref{eqn_dY}.

\end{assumption}

\begin{assumption}
\label{assumption_normal_mu0}
$\mu_0$ is normally distributed with mean $m_1(0)$ and variance $v_1(0)$, $\mu_0 \sim \Nc(m_1(0),v_1(0))$, and is independent of $(W_t)_{t \ge 0}$ and $(\bar{W}_t)_{t \ge 0}$.\footnote{$\Nc(m,v)$ denotes a normal distribution with mean $m$ and variance $v$.}
\end{assumption}

\subsection{Utility Maximization for $\Cc_1$ and $\Cc_2$ Strategies}
\label{subsec_OU_utility}

In this section we solve Problem \ref{prob_utility} for strategies $f \in \Cc_1$ and $f \in \Cc_2$ when the dynamics of $\mu$ are given by \eqref{eqn_OU_drift}.
The general solutions to such a problem are obtained previously in Theorem \ref{theorem_general_utility_C1} for $\Cc_1$ strategies and in Theorem \ref{theorem_general_utility_C2} for $\Cc_2$ strategies, respectively.
Here when we make assumptions on the dynamics of $\mu$ and the distribution of $\mu_0$, we can further reduce the general results into fully explicit forms, see  Theorem \ref{thm_OU_utility_linear} for $\Cc_1$ strategies and Theorem \ref{thm_OU_utility_general} for $\Cc_2$ strategies.

 \begin{theorem}
	\label{thm_OU_utility_linear}
	Let Assumptions \ref{assumption_OU} and \ref{assumption_normal_mu0} hold, then 
	the optimal ExpMA strategy $\mathsf{f}_1^{*} \in \Cc_1$ to Problem \ref{prob_utility} is 
	given by $\mathsf{f}_1^{*}(t,z) = a_1^* \cdot z + b_1^*$, where $a_1^*$ and $b_1^*$ are given by  \eqref{eqn_utility_C1_a1b1} in Theorem \ref{theorem_general_utility_C1}. Furthermore,  
	$A(T)$, $B(T)$, $C(T)$, and $D(T)$, defined by \eqref{eqn_ABCD}, are computed explicitly by \eqref{eqn_OU_A}-\eqref{eqn_OU_D}.
\end{theorem}

\begin{proof}  
	The detailed computations of  $A(T)$, $B(T)$, $C(T)$, and $D(T)$ are given in Appendix \ref{sec:proof-1}.
\end{proof}

In the above theorem, we consider Problem \ref{prob_utility} for the class of affine functionals.
Next, we extend the analysis to a larger class (square-integrable functionals), and present the results in Theorem \ref{thm_OU_utility_general}. The following notations are needed. 
\begin{align}
m_1(t) := \Eb[\mu_t], && v_1(t) := \Vb[\mu_t], && m_2(t):= \Eb[Z_t], && v_2(t) := \Eb[Z_t], 
&& m_3(t) := \Eb[\mu_t Z_t].
\end{align}
Explicit expressions for the above quantities are given respectively in equations \eqref{eqn_drift_mean},  \eqref{eqn_drift_variance}, \eqref{eqn_OU_m2}, \eqref{eqn_OU_v2}, and \eqref{eqn_OU_m3} in  Appendix \ref{sec:proof-1}.

\begin{theorem}
\label{thm_OU_utility_general}
Let Assumptions \ref{assumption_OU} and \ref{assumption_normal_mu0} hold, then 
the optimal ExpMA strategy $\mathsf{f}_2^* \in \Cc_2$ to Problem \ref{prob_utility} is
\begin{align}
\label{eqn_OU_utility_general_optimal}
\mathsf{f}_2^*(t,z)=a_2^*(t) \cdot z+ b_2^*(t),
\end{align}
where
\begin{align}\label{eqn_opt_a2b2}
  a_2^*(t)=\frac{m_3(t)-m_1(t)m_2(t)}{v_2(t)\sigma^2}\quad \text{  and  } \quad 
  b_2^*(t)= \frac{m_1(t)}{\sigma^2}-\frac{m_2(t)\left(m_3(t)-m_1(t)m_2(t)\right)}{v_2(t)\sigma^2}.
\end{align}
\end{theorem}

\begin{proof}
The proof is similar to that of Theorem \ref{theorem_general_utility_C2} and hence is omitted.
\end{proof}

\begin{remark}
In both Theorems \ref{thm_OU_utility_linear} and \ref{thm_OU_utility_general}, the optimal ExpMA strategy is obtained in closed-form, i.e., once the model parameters in \eqref{eq:dS} and \eqref{eqn_OU_drift} are estimated or given, we are able to compute $(a_1^*, b_1^*)$ using  \eqref{eqn_utility_C1_a1b1} and $(a_1^*(t), b_1^*(t))$ using \eqref{eqn_opt_a2b2}, respectively. 
Theorem \ref{thm_OU_utility_general} shows that under a more general class $\Cc_2$, the optimal ExpMA strategy $\mathsf{f}_2^*$ is still in affine form.
Such a strong result cannot be deduced from the general solution in Theorem \ref{theorem_general_utility_C2}.
\end{remark}

\begin{remark}
If the drift $\mu$ in \eqref{eqn_OU_drift} is fully observable, \cite{kim1996dynamic} provide analytical solutions to Problem \ref{prob_utility}. 
Specifically, the value function in their studies is
\begin{align}\label{eqn_def_optimal_value_obs}
\bar V(T)  := \sup_{\pi \in \bar \Ac} \Eb \[\ln\frac{\Pi_T^\pi}{\Pi_0}\],
\end{align}
where 
\begin{align} 
\label{eqn_def_Ac_bar}
\bar \Ac := \left\{ \pi \text{ is } \Fb\text{-adapted}  \, \Bigg| \, \Eb \left[\int_0^T \pi_t^2 \dd t \right] < \infty \right\}.
\end{align}
They obtain the optimal investment strategy as $\bar{\pi}_t^* = \mu_t / \sigma^2$ for all $t \in [0,T]$. 

If the drift $\mu$ in \eqref{eqn_OU_drift} is unobservable, \cite{lakner1995utility, lakner1998optimal} considers the problem 
\begin{equation}\label{e112}
\check{V}(T) :=\sup_{\pi \in \Ac^S} \; \Eb  \big[ U(\Pi_T^\pi) \big],
\end{equation}
where 
\begin{equation}\label{e111}
\Ac^S:=\left\{\pi \text{ is } \Fb^S\text{-adapted}\ \bigg|\  \Eb \left[\int_0^T \pi_t^2 \, \dd t \right]<\infty\right\}.
\end{equation} 
Given $U(x)=\ln(x)$, the optimal investment strategy is obtained by 
\begin{align}
\label{eqn_partial_OU_strategy}
\check{\pi}^*_t= \dfrac{\Eb \left[ \mu_t|\mathcal{F}_t^S \right] }{\sigma^2}.
\end{align}

Notice that $\bar \Ac$ in Problem \eqref{eqn_def_optimal_value_obs} is not the same as $\Ac^S$  in Problem \eqref{e112}, where $\pi$ is $\Fb^S$-adapted. 
As seen above, when the drift is unobservable, the best strategy, among all that are adapted to the filtration generated by the price process, is to use the true filter $\Eb[\mu_t|\mathcal{F}_t^S]$ to replace the drift $\mu_t$ at all times.
\end{remark}

In the remaining of this subsection, we compare the value functions to the logarithmic utility maximization under $\Cc_1$, $\Cc_2$, $\bar \Ac$- and $\Ac^S$-adapted strategies.
To this purpose, denote
\begin{align}
 V_1^*(T) : = \sup_{f \in \Cc_1} V^f(T) := \sup_{f \in \Cc_1} \Eb \left[ \ln \left( \frac{\Pi^f_T}{\Pi_0} \right) \right] , \qquad
  V_2^*(T) : = \sup_{f \in \Cc_2} V^f(T) := \sup_{f \in \Cc_2} \Eb \left[ \ln \left( \frac{\Pi^f_T}{\Pi_0} \right) \right].
\end{align}

Since $\Cc_1 \subset \Cc_2$, $V_1^*(T) \le V_2^*(T)$ for all $T >0$.
Furthermore, $\mathsf{f}_2^* \notin \Cc_1$ but $\mathsf{f}_1^* \in \Cc_2$, we claim $  V_1^*(T) < V_2^*(T)$ for all $T >0$. 
In Section \ref{subsec_general_solutions}, we have computed $V^f(T)$ as $g(a,b;T)$ when $f \in \Cc_1$, see \eqref{eqn_g}. 
In consequence, $V_1^*(T)$ defined above is equal to $g(a_1^*, b_1^*;T)$, where $a_1^*$ and $b_1^*$ are given by \eqref{eqn_utility_C1_a1b1}. 
The proposition below presents the comparison results among $\bar{V}(T)$,  $\check{V}(T)$ and $V_2^*(T)$.

\begin{proposition}
Let Assumptions \ref{assumption_OU} and \ref{assumption_normal_mu0} hold, we have 
\begin{align}
\bar{V}(T) > V_2^*(T) > \check{V}(T) \qquad \text{for all } T>0.
\end{align}
\end{proposition}

\begin{proof}
	
	By plugging the optimal ExpMA strategy $\mathsf{f}_2^*$, given by \eqref{eqn_OU_utility_general_optimal}, into the above expression for $f$, we obtain
	\begin{align}
	V_2^*(T)&=\int_0^T \Eb \left[\left( \mu_t \cdot \mathsf{f}_2^*(t,Z_t) - \frac{1}{2}\sigma^2 \cdot (\mathsf{f}_2^*(t,Z_t))^2 \right) \right] \dd t\\
	&=\int_0^T \left[ \frac{(m_3(t)-m_1(t)m_2(t))^2}{2v_2(t)\sigma^2}+\frac{m_1^2(t)}{2\sigma^2} \right]\, \dd t \\
	&=\frac{1}{2\sigma^2}\int_0^T \left[\text{corr}^2(Z_t,\mu_t)\cdot  v_1(t)+m_1^2(t) \right]\, \dd t, \label{eqn_value_function_f2}
	\end{align}
	where $\text{corr}(Z_t, \mu_t)$ is the correlation coefficient between $Z_t$ and $\mu_t$.
	
	Assuming $\mu$ is
	observable, the optimal strategy $\bar{\pi}^*$ to Problem \eqref{eqn_def_optimal_value_obs} is
	$\bar{\pi}^*_t = \mu_t/\sigma^2$,
	and then
	\begin{align}
	\label{eqn_V_bar}
	\bar V(T)=\Eb \left[\int_0^T \left( \mu_t \bar{\pi}^*_t - \frac{1}{2}\sigma^2 (\bar{\pi}^*_t)^2 \right)\dd t\right]=\frac{1}{2\sigma^2}\int_0^T\[v_1(t)+m_1^2(t)\]  \dd t.
	\end{align}

	Recall the results above, the value function $\check{V}(T)$ is achieved when $\pi^*(t) = \Eb[\mu_t | \mathcal{F}^S_t]/\sig^2$. Since $W$, $\bar{W}$, and $\mu_0$ are independent due to Assumptions \ref{assumption_OU} and \ref{assumption_normal_mu0}, we compute
	\begin{align}
\Eb[\mu_t | \mathcal{F}^S_t]= \Eb[\mu_t] = \bar{\mu} + (m_1(0) - \bar{\mu}) e^{- \kappa t},
	\end{align}  
	where we have used the dynamics of $\mu$ in \eqref{eqn_OU_drift} 
	to derive the last equality. 
	Using this result, we are able to obtain $\check{V}(T)$ as
	\begin{align}
	\check{V}(T) 
	& = \frac{1}{2 \sig^2} \int_0^T m_1^2(t) \dd t . \label{eqn_V_check}
	\end{align}
	
Since $0 < \text{corr}(Z_t, \mu_t) < 1$ and $v_1(t) = \Vb[\mu_t] >0$ for all $0 \le t \le T$, the comparison results are then obtained using \eqref{eqn_value_function_f2}, \eqref{eqn_V_bar} and \eqref{eqn_V_check}.
\end{proof}

\subsection{Long-term Growth Rate Maximization for $\Cc_1$ and $\Cc_2$ Strategies}
\label{subsec_OU_longrun}
In this section, we study Problem \ref{prob_growth} for strategies $f \in \Cc_1$ and $f \in \Cc_2$ when the dynamics of $\mu$ are given by \eqref{eqn_OU_drift}. The main results are presented in Theorem \ref{thm_OU_growth}.

We begin our analysis by noticing that, as $t\rightarrow\infty$, we have
\begin{align}
a_\infty &:= \lim_{t \to \infty} a_2^*(t) 
=\frac{\lambda \delta^2}{\sigma^2} \cdot \frac{1}{\kappa(\kappa+\lambda)\sigma^2 + \delta^2}, \label{eqn_def_a_infty} \\
b_\infty &:= \lim_{t \to \infty} b_2^*(t) = \frac{\bar \mu}{\sigma^2} - \frac{2\bar\mu-\sigma^2}{2\lambda} \cdot a_\infty, \label{eqn_def_b_infty}
\end{align}
where $a_2^*(t)$ and $b_2^*(t)$ are given by \eqref{eqn_opt_a2b2}.

Define $\mathsf{f}_\infty$ by
\begin{align}
\mathsf{f}_\infty(z) := a_\infty \cdot z + b_\infty, \label{eqn_def_f_infty}
\end{align}
where $a_\infty$ and $b_\infty$ are defined by \eqref{eqn_def_a_infty} and \eqref{eqn_def_b_infty}, respectively. It is clear that $\mathsf{f}_\infty \in \Cc_1 \subset \Cc_2$.

Define $\eta:=\eta(\lambda)$ by
\begin{align}\label{eqn_def_eta}
\eta=\eta(\lambda) = \eta(\lambda; \kappa, \bar{\mu}, \sigma, \delta) := 
\frac{\delta^4}{4\kappa \sigma^2} \cdot \frac{\lambda}{\kappa \sigma^2 (\kappa+\lambda)^2 + (\kappa + \lambda)\delta^2}+ \frac{\bar\mu}{2 \sigma^2}.
\end{align}
We have the following result.

\begin{theorem}\label{thm_OU_growth}
Let Assumptions \ref{assumption_OU} and \ref{assumption_normal_mu0} hold, 
we have
\begin{equation}\label{e12}
\lim_{T \to \infty} \frac{1}{T} V_1^*(T)=\lim_{T \to \infty} \frac{1}{T} V_2^*(T)=\lim_{T \to \infty} \frac{1}{T} \Eb \left[ \ln \left( \frac{\Pi^{\mathsf{f}_\infty}_T}{\Pi_0} \right) \right]  = \eta.
\end{equation}
In particular, the above result implies that
\begin{equation}\label{e13}
\lim_{T \to \infty} \frac{1}{T} \Eb \left[ \ln \left( \frac{\Pi^{\mathsf{f}_\infty}_T}{\Pi_0} \right) \right]=\sup_{f\in\Cc_i}\lim_{T \to \infty} \frac{1}{T} \Eb \left[ \ln \left( \frac{\Pi^f_T}{\Pi_0} \right) \right], \quad i=1,2.
\end{equation}
That is, $\mathsf{f}_\infty(z)$, given by \eqref{eqn_def_f_infty}, is an optimal ExpMA strategy to 
Problem \ref{prob_growth} within both the $\Cc_1$ class and the $\Cc_2$ class.
\end{theorem}

\begin{proof}
	Obviously we have that
	\begin{equation}\label{e11}
	\Eb \left[ \ln \left( \frac{\Pi^{\mathsf{f}_2^*}_T}{\Pi_0} \right) \right]=V_2^*(T)\geq\Eb\left[ \ln \left( \frac{\Pi^{\mathsf{f}_1^*}_T}{\Pi_0} \right) \right]= V_1^*(T)\geq\Eb \left[ \ln \left( \frac{\Pi^{\mathsf{f}_\infty}_T}{\Pi_0} \right) \right].
	\end{equation}
	Moreover,
	\begin{align}
	&\lim_{T \to \infty} \frac{1}{T} \Eb \left[ \ln \left( \frac{\Pi^{\mathsf{f}_\infty}_T}{\Pi_0} \right) \right]
	=\lim_{T\rightarrow\infty}\frac{1}{T}\Eb\[\int_0^T\[\mu_t(a_\infty \cdot Z_t+b_\infty)-\frac{\sigma^2}{2}\(a_\infty^2 \cdot Z_t^2+2a_\infty b_\infty \cdot Z_t+b_\infty^2\)\]\, \dd t\]\\
	=&\lim_{T\rightarrow\infty}\frac{1}{T}\int_0^T\[a_\infty m_3(t)+b_\infty m_1(t)-\frac{\sigma^2}{2}\(a_\infty^2(m_2^2(t)+v_2(t))+2a_\infty b_\infty m_2(t)+b_\infty^2\)\]\, \dd t\\
	=&\lim_{T\rightarrow\infty}\frac{1}{T}\Eb\[\int_0^T\[\mu_t(a_2^*(t) Z_t+b_2^*(t))-\frac{\sigma^2}{2}\( (a_2^*(t))^2Z_t^2+2a_2^*(t) b_2^*(t) Z_t+ (b_2^*(t))^2\)\]\, \dd t\]\\
	=&\lim_{T\rightarrow\infty}\frac{1}{T}\Eb \left[\ln \frac{\Pi_T^{\mathsf{f}_2^*}}{\Pi_0} \right]
	=\eta,
	\end{align}
	where the third equality follows from \eqref{eqn_def_a_infty} and \eqref{eqn_def_b_infty},
	and the last equality follows from \eqref{eqn_value_function_f2}.
	This together with \eqref{e11} implies \eqref{e12}.
	
	Since
	\begin{align}
	\lim_{T \to \infty} \frac{1}{T} V_i^*(T)=\lim_{T \to \infty} \frac{1}{T} \Eb \left[ \ln \left( \frac{\Pi^{\mathsf{f}_\infty}_T}{\Pi_0} \right) \right]\leq\sup_{f\in\Cc_i}\lim_{T \to \infty} \frac{1}{T} \Eb \left[ \ln \left( \frac{\Pi^f_T}{\Pi_0} \right) \right]\leq\lim_{T \to \infty} \frac{1}{T} V_i^*(T), \quad i=1,2,
	\end{align}
	we have \eqref{e13} holds.
\end{proof}

By Theorem \ref{thm_OU_growth}, $\lim_{T \to \infty} \frac{1}{T} V_2^*(T)$ is equal to $\eta$, which is defined by \eqref{eqn_def_eta} and solely depends on the moving average constant $\lambda$, once the model parameters $\kappa$, $\delta$, $\bar{\mu}$, and $\sigma$ are fixed.
The next proposition provides an upper bound for $\eta(\lambda)$.

\begin{proposition}
\label{prop_optimal_lambda}
Let Assumptions \ref{assumption_OU} and \ref{assumption_normal_mu0} hold, 
we have
\begin{align}
\eta(\lambda)\leq\frac{\delta^2}{4\sigma^2\kappa}\cdot\frac{\delta^2}{2\sigma\kappa\sqrt{\sigma^2\kappa^2+\delta^2}+2\sigma^2\kappa^2+\delta^2}+\frac{\bar\mu^2}{2\sigma^2},
\end{align}
where the equality holds if and only if 
\begin{align}
\lam = \hat\lam :=\sqrt{\kappa^2+\frac{\delta^2}{\sigma^2}}.
\end{align}
\end{proposition}

\begin{proof}
From \eqref{eqn_def_eta}, we compute
\begin{align}
\frac{\partial \eta}{\partial \lambda} = - \frac{ \kappa \delta^4 \sigma^2}{4\sigma^2\kappa} \cdot \frac{\lambda^2 - \hat{\lambda}^2}{ \left[\sigma^2\kappa(\kappa+\lambda)^2+\delta^2(\kappa+\lambda) \right]^2},
\end{align}
where $\hat{\lambda}$ is defined above. 
The desired upper bound is obtained when $\lambda$ is replaced by $\hat{\lambda}$ in  \eqref{eqn_def_eta}.
\end{proof}

Next, we compare the limit behavior of $V_2^*(T)$ with that of $\bar{V}(T)$  and $\check{V}(T)$, which are defined respectively by \eqref{eqn_def_optimal_value_obs} and \eqref{e112}. 
We present the comparison results below.

\begin{proposition}
\label{prop_expected_value_comparison}
Let Assumptions \ref{assumption_OU} and \ref{assumption_normal_mu0} hold, 
we have 
\begin{align}
\lim\limits_{T \to \infty} \frac{1}{T} \check{V}(T) = \frac{\bar{\mu}^2}{2 \sig^2} < 
\lim_{T \to \infty} \frac{1}{T} V_2^*(T) = \eta < \lim_{T \to \infty} \frac{1}{T} \bar V(T) =  \xi,
\end{align}
where $\eta$ and $\xi$ are defined by \eqref{eqn_def_eta} and \eqref{eqn_def_xi}, respectively.

In addition, for $\lambda,\sigma,\delta>0$ and $\bar\mu\neq 0$, we have
\begin{align}
\lim_{\kappa\rightarrow 0} \frac{\eta(\lambda;\kappa,\bar{\mu},\delta,\sigma)}{\xi(\kappa,\bar{\mu},\delta,\sigma)} &= \lim_{\kappa\rightarrow\infty}\frac{\eta(\lambda;\kappa,\bar{\mu},\delta,\sigma)}{\xi(\kappa,\bar{\mu},\delta,\sigma)}=1, \\
\lim_{\delta\rightarrow 0}\frac{\eta(\lambda;\kappa,\bar{\mu},\delta,\sigma)}{\xi(\kappa,\bar{\mu},\delta,\sigma)} &= \lim_{\delta\rightarrow\infty}\frac{\eta(\hat\lambda(\kappa,\delta,\sigma);\kappa,\bar{\mu},\delta,\sigma)}{\xi(\kappa,\bar{\mu},\delta,\sigma)}=1.
\end{align}
\end{proposition}

\begin{proof}
Recall the values functions $\bar{V}(T)$  and $\check{V}(T)$ are computed explicitly in  \eqref{eqn_V_bar} and \eqref{eqn_V_check}.
Since $\lim_{t \to \infty} v_1(t) = \frac{\delta^2}{2\kappa}$ and $\lim_{t \to \infty} m_1(t) = \bar \mu$,  taking the limits leads to 
\begin{align}
\label{eqn_def_xi}
\lim_{T\rightarrow\infty} \frac{1}{T} \bar V(T) &=\frac{\delta^2}{4\kappa\sigma^2}+\frac{\bar\mu^2}{2\sigma^2}=:\xi(\kappa,\bar{\mu},\delta,\sigma), \\
\lim\limits_{T \to \infty} \frac{1}{T} \check{V}(T) &= \frac{\bar{\mu}^2}{2 \sig^2}.
\end{align}		
The above comparison inequalities are immediate results of  Proposition \ref{prop_optimal_lambda}. 
\end{proof}

\begin{remark}
Proposition \ref{prop_expected_value_comparison} shows that the long-term growth rate loss, due to partial information on the drift process,  is strictly greater than 0, i.e., $\xi - \eta >0$. 
However, if $\kappa$ or $\delta$ approaches 0 or $\infty$, such a loss is asymptotically negligible. 
In addition, we have $\eta > \bar{\mu}^2/(2 \sig^2)$, implying that the optimal ExpMA strategy achieves greater long-term growth rate comparing to the optimal $\Fb^S$-adapted strategy.
\end{remark}

\section{Analysis for the Case of a Two-State Markov Drift}
\label{sec_MC}

In this section, We solve Problems \ref{prob_utility} and \ref{prob_growth} when the drift is given by a Markov chain, which we specify in Assumption \ref{assumption_MC}. 
Key findings are summarized in Theorems \ref{thm_MC_utility_C1},  \ref{thm_MC_growth_C1}, \ref{thm_MC_utility_C2}, and \ref{thm_MC_growth_C2}.

\begin{assumption}
\label{assumption_MC}
The drift $\mu$ is modeled by a time-homogeneous two-state CTMC, which is independent of the Brownian motion $W$. 
Furthermore, suppose:
\begin{itemize}
\item The state space of $\mu$ is $\{\rho_1,\rho_2\}$, where $\rho_1$ and $\rho_2$ are two constants such that $\rho_1<\rho_2$ (i.e., $\mu$ jumps between $\rho_1$ and $\rho_2$).

\item The generator matrix of $\mu$ is given by
\begin{align}
G=\begin{bmatrix}
-\alpha&\alpha\\
\beta&-\beta
\end{bmatrix},
\end{align}
where $\alpha, \, \beta>0$.
\end{itemize}
We impose a technical condition\footnote{Such a technical assumption is necessary for $n_4(t)$ in \eqref{eqn_n4t} to be well-defined.}: $\lambda \neq \alpha + \beta$, where $\lambda$ is the exponential moving average constant, see \eqref{eqn_dY}.
\end{assumption}
\vspace{2ex}

Denote by $P(t)=[P_{ij}(t)]_{i,j=1,2}$ the transition matrix of the drift $\mu$. That is,
\begin{align}
P_{ij}(t):=\Pb (\mu_t=\rho_j\,|\,\mu_0=\rho_i ), \qquad i,j =1,2.
\end{align}
It is easy to verify that
\begin{align}
P(t)=e^{tG}=
\begin{bmatrix}
\frac{\beta}{\alpha+\beta}+\frac{\alpha}{\alpha+\beta}e^{-(\alpha+\beta)t},&\frac{\alpha}{\alpha+\beta}-\frac{\alpha}{\alpha+\beta}e^{-(\alpha+\beta)t}\\[2ex]
\frac{\beta}{\alpha+\beta}-\frac{\beta}{\alpha+\beta}e^{-(\alpha+\beta)t},&\frac{\alpha}{\alpha+\beta}+\frac{\beta}{\alpha+\beta}e^{-(\alpha+\beta)t}
\end{bmatrix}.
\end{align}

\begin{assumption}
\label{assumption_MC_mu0}
$\mu_0$ has the stationary distribution of the CTMC, namely,
\begin{align}
\Pb (\mu_0=\rho_1)=\frac{\beta}{\alpha+\beta} \quad \text{ and } \quad 
\Pb (\mu_0=\rho_2)=\frac{\alpha}{\alpha+\beta}.
\end{align}
\end{assumption}
\vspace{2ex}

If Assumption \ref{assumption_MC_mu0} holds true, $\mu_t$ has the same distribution as $\mu_0$ for all $t\geq 0$. Denote by $n_1$ the expected value of $\mu_t$. We obtain
\begin{align}
\label{eqn_n1}
n_1:=\Eb[\mu_t]=\frac{\beta}{\alpha+\beta} \cdot \rho_1+\frac{\alpha}{\alpha+\beta}\cdot \rho_2 \, .
\end{align}

\subsection{Analysis on $\Cc_1$ Strategies}
\label{subsec_MC_C1}

In this section, we seek solutions to Problems \ref{prob_utility} and \ref{prob_growth} for $\Cc_1$ strategies
when the drift $\mu$ is given by the CTMC described above.
The solutions to Problems \ref{prob_utility} and \ref{prob_growth} are given respectively in Theorems \ref{thm_MC_utility_C1} and \ref{thm_MC_growth_C1}.

Recall that the general solution to Problem \ref{prob_utility} is found in Theorem \ref{theorem_general_utility_C1}. Now under the Markovian assumptions, we obtain explicit formulas for $A(T)$, $B(T)$, $C(T)$ and $D(T)$ defined in \eqref{eqn_ABCD}. We introduce the following notations and then present the results:
\begin{align}
n_1:= \Eb[\mu_t], \quad n_2(t) :=\Eb[Z_t], \quad n_3(t):=\Eb[\mu_tZ_t], \quad \text{and} \quad n_4(t):=\Eb[Z_t^2], \quad t \ge 0,
\end{align}
where $n_1$ is computed in \eqref{eqn_n1}.

\begin{theorem}
\label{thm_MC_utility_C1}

Let Assumptions \ref{assumption_MC} and \ref{assumption_MC_mu0} hold, then 
the optimal ExpMA strategy $\mathfrak{f}_1^{*} \in \Cc_1$ to Problem \ref{prob_utility} is 
given by $\mathfrak{f}_1^*(t,z) = a_1^* \cdot z + b_1^*$, where $a_1^*$ and $b_1^*$ are given by \eqref{eqn_utility_C1_a1b1} in Theorem \ref{theorem_general_utility_C1}.
In addition, 
we obtain
\begin{align}
n_2(t) &= \Eb[Z_t] = \frac{1}{\lambda} \left(n_1-\frac{1}{2}\sigma^2\right) \left(1-e^{-\lambda t}\right), \label{eqn_n2t}\\
n_3(t) &= \Eb[\mu_tZ_t] = \left(\frac{n_1^2}{\lambda}-\frac{n_1\sigma^2}{2\lambda}\right)\left(1-e^{-\lambda t}\right)+\frac{\gamma}{\alpha+\beta+\lambda}\left(1-e^{-(\alpha+\beta+\lambda)t}\right), \label{eqn_n3t}\\
n_4(t) &= \Eb[Z_t^2] =  \frac{2\gamma}{(\lambda-\alpha-\beta)(\lambda+\alpha+\beta)}\left(1-e^{-(\alpha+\beta+\lambda)t}\right) + \left[\frac{\sigma^2}{2\lambda}-\frac{\gamma}{\lambda(\lambda-\alpha-\beta)} \right] \left(1-e^{-2\lambda t}\right) \\
& \hspace{10ex} + \frac{1}{\lambda^2} \left(n_1 - \frac{\sigma^2}{2} \right)^2 \left(1-e^{-\lambda t}\right)^2, \label{eqn_n4t}
\end{align}
where $\gamma := \Vb[\mu_t] = \frac{\alpha \beta}{(\alpha + \beta)^2} (\rho_1 - \rho_2)^2$. 
\end{theorem}

\begin{proof}
Please refer to Appendix \ref{appen_thm_MC_utility_C1} for the computation of $n_i(t)$, where $i=1,2,3,4$.
\end{proof}

Next, we turn our attention to Problem \ref{prob_growth} for strategies $f \in \Cc_1$ when the drift $\mu$ is modeled by a CTMC.
We begin our analysis by observing that
\begin{align}
h_\infty &:=\lim_{T\rightarrow\infty}\frac{A(T)}{T}=\frac{n_1^2}{\lambda}-\frac{n_1\sigma^2}{2\lambda}+\frac{\gamma}{\lambda+\alpha+\beta},\\
i_\infty &:=\lim_{T\rightarrow\infty}\frac{C(T)}{T}=\frac{\gamma}{\lambda(\lambda+\alpha+\beta)}+\frac{\sigma^2}{2\lambda}+\left(\frac{n_1}{\lambda}-\frac{\sigma^2}{2\lambda}\right)^2,\\
j_\infty &:=\lim_{T\rightarrow\infty}\frac{D(T)}{T}=\frac{n_1}{\lambda}-\frac{\sigma^2}{2\lambda}.
\end{align}

Recall from Theorem \ref{thm_MC_utility_C1} that the optimal strategy is $\mathfrak{f}_1^*(t,z) = a_1^* \cdot z + b_1^*$, and $a_1^*$ and $b_1^*$ are both constants which depend on the time horizon $T$. 
Here, to emphasize such dependence, we write them as $a_1^*(T)$ and $b_1^*(T)$.
Immediately, we deduce that
\begin{align}
c_\infty &:=\lim_{T  \rightarrow\infty} a_1^*(T)=\frac{2\lambda \gamma}{2 \gamma \sigma^2+\sigma^4(\lambda + \alpha+\beta)}, \label{eqn_def_c_infty} \\
d_\infty&:=\lim_{T \rightarrow\infty} b_1^*(T)=\frac{\gamma+n_1(\lambda+\alpha+\beta)}{2 \gamma+\sigma^2(\lambda  + \alpha+\beta)}. \label{eqn_def_d_infty}
\end{align}

With the above limiting results, we present the solution to Problem \ref{prob_growth} for strategies $f \in \Cc_1$ as follows.

\begin{theorem}
\label{thm_MC_growth_C1}
Let Assumptions \ref{assumption_MC} and \ref{assumption_MC_mu0} hold, we have
\begin{align}
\lim_{T \to \infty} \frac{1}{T} \Eb \left[ \ln \left( \frac{\Pi^{\mathfrak{f}_1^*}_T}{\Pi_0} \right) \right]
=\lim_{T \to \infty} \frac{1}{T} \Eb \left[ \ln \left( \frac{\Pi^{\mathfrak{f}_\infty}_T}{\Pi_0} \right) \right] =\mathfrak{g}(c_\infty,d_\infty)>0,
\end{align}
where $\mathfrak{g}(\cdot,\cdot)$ is defined by
\begin{align}
\mathfrak{g}(x,y):=h_\infty x + n_1 y - \frac{1}{2}\sigma^2\left(i_\infty x^2+2j_\infty xy + y^2\right), \quad \forall \, x, y \in \mathbb{R}.
\end{align}
The optimal ExpMA strategy to Problem \ref{prob_growth} 
for strategies $f \in \Cc_1$ is
\begin{align}
\mathfrak{f}_\infty(t, z)=c_\infty \cdot z+d_\infty,
\end{align}
where $c_\infty$ and $d_\infty$ are defined by \eqref{eqn_def_c_infty} and \eqref{eqn_def_d_infty}, respectively.
\end{theorem}

\begin{proof}
The proof is similar to that of Theorem \ref{thm_OU_growth}, and hence is omitted. 
\end{proof}

\subsection{Analysis on $\Cc_2$ Strategies}
\label{subsec_MC_C2}

In this section, we extend our analysis from $\Cc_1$ (affine strategies) to a larger class $\Cc_2$ (square-integrable strategies).
The main results are Theorems \ref{thm_MC_utility_C2} and \ref{thm_MC_growth_C2}, where we provide solutions to Problems \ref{prob_utility} and \ref{prob_growth}, respectively.
We revisit Problem \ref{prob_utility} for $\Cc_2$ strategies, and provide explicit characterizations to the optimal strategy $\mathfrak{f}_2^*$ in the following theorem.

\begin{theorem}
\label{thm_MC_utility_C2}
Let Assumptions \ref{assumption_MC} and \ref{assumption_MC_mu0} hold, 
the optimal ExpMA strategy $\mathfrak{f}_2^*$ in $\Cc_2$ to Problem \ref{prob_utility} is given by 
\begin{align}
\mathfrak{f}_2^*(t, Z_t) = \frac{1}{\sigma^2} \, \Eb [ \mu_t \, | \, Z_t] = 
\frac{1}{\sigma^2} \, \Eb [ \mu_0 \, | \, Q_t],
\end{align}
where $Q_t$ is defined by 
\begin{align}
\label{eqn_def_Qt}
Q_t:= \underbrace{\int_0^te^{-\lambda s}\mu_s\, \dd s}_{:=Q_{1,t}} + 
\underbrace{\sigma\int_0^te^{-\lambda s}\,\dd W_s-\frac{\sigma^2}{2\lambda}(1-e^{-\lambda t})}_{:=Q_{2,t}}
\end{align}
and $\Eb [ \mu_0 \, | \, Q_t]$ is calculated by \eqref{eq10}.
\end{theorem}

\begin{proof}

The first result $\mathfrak{f}_2^*(t,Z_t)=\Eb[\mu_t|Z_t]/\sigma^2$ is a direct consequence of the general solution from Theorem \ref{theorem_general_utility_C2}.
Recall that $Z_t$ can be rearranged as
\begin{align}
Z_t = \int_0^t \mu_s e^{-\lambda (t-s)} \, \dd s  + \sigma \int_0^t e^{-\lambda(t-s)}\,\dd W_s - \frac{\sigma^2}{2 \lambda} \left(1-e^{-\lambda t} \right). 
\end{align}
Since the drift $\mu$ is a two-state stationary CTMC, it is reversible. This observation together with the reversibility of Brownian motion implies  that 
\begin{align}
\left(\mu_0,\; \int_0^te^{-\lambda s}\mu_s\, \dd s, \; \sigma\int_0^te^{-\lambda s}\, \dd W_s \right) \quad \text{ and } \quad 
\left(\mu_t,\; \int_0^t\mu_se^{-\lambda(t-s)}\,\dd s, \; \sigma \int_0^t e^{-\lambda(t-s)}\,\dd W_s \right)
\end{align} have the same joint distribution, and thus 
\begin{align}
\Eb[\mu_t|Z_t] = \Eb[\mu_0|Q_t].
\end{align}

Let $u(t,\cdot)$ and $v(t,\cdot)$ be the conditional cumulative distribution functions (c.d.f.) of $Q_{1,t}$ given $\mu_0=\rho_1$ and $\mu_0=\rho_2$ respectively, 
where $Q_{1,t}$ is defined by \eqref{eqn_def_Qt}. 
That is
\begin{align}
u(t,x)
	&:=	\Pb( Q_{1,t} \leq x | \mu_0 = \rho_1 ) , &
v(t,x)
	&:=	 \Pb( Q_{1,t} \leq x  | \mu_0 = \rho_2 ) .  \label{def:uv}
\end{align}
As $\frac{\rho_1}{\lambda}(1-e^{-\lambda t}) \le Q_{1,t} \le \frac{\rho_2}{\lambda}(1-e^{-\lambda t})$, we obviously have
\begin{align}
u(t,x)=v(t,x)=0\ \text{ if }\ x<\frac{\rho_1}{\lambda}(1-e^{-\lambda t})\quad \text{ and }\quad u(t,x)=v(t,x)=1\ \text{ if }\ x>\frac{\rho_2}{\lambda}(1-e^{-\lambda t}).
\end{align}
Denote by $\phi(t,x)$ the probability density function (p.d.f.) of $Q_{2,t}$, defined by \eqref{eqn_def_Qt}, i.e.,
\begin{align}
\phi(t,x)=\sqrt{\frac{\lambda}{\pi\sigma^2\left(1-e^{-2\lambda t}\right)}}\cdot \exp \[ -\dfrac{\left[x+\frac{\sigma^2}{2\lambda}\left(1-e^{-\lambda t}\right)\right]^2}{\frac{\sigma^2}{\lambda}\left(1-e^{-2\lambda t}\right)}\] , \quad x \in (-\infty, \infty).
\end{align}
Then the conditional c.d.f. of $Q_t$ given $\mu_0=\rho_1$ is 
\begin{align}
F_{Q_t \, | \, \mu_0 = \rho_1} (x) = \int_{-\infty}^\infty u(t,z) \cdot \phi(t,x-z)\, \dd z.
\end{align}
Using the dominated convergence theorem, we obtain the conditional p.d.f. of $Q_t$ given $\mu_0=\rho_1$ by
\begin{align}
p(t,x):=\int_{-\infty}^\infty u(t,z)\cdot \phi'(t, x-z)\, \dd z,
\end{align}
where $\phi'(t,x)= \frac{\partial \phi(t,x)}{\partial x}$.
Similarly, the conditional p.d.f. of $Q_t$ given $\mu_0=\rho_2$ is obtained by
\begin{align}
q(t,x):=\int_{-\infty}^\infty v(t,z) \cdot \phi'(t, x-z) \, \dd z.
\end{align}
Using $p(t,x)$ and $q(t,x)$, and the distribution of $\mu_0$ in Assumption \ref{assumption_MC_mu0}, we have
\begin{align}
\Pb (\mu_0=\rho_1\,|\,Q_t=x) 
=\frac{\beta \cdot p(t,x)}{\beta \cdot p(t,x)+\alpha \cdot q(t,x)}
\quad \text{ and }\quad  
\Pb (\mu_0=\rho_2\,|\,Q_t=x) = \frac{\alpha \cdot q(t,x)}{\beta \cdot p(t,x)+\alpha \cdot q(t,x)}.
\end{align}

Therefore, we obtain
\begin{equation}\label{eq10}
\Eb[\mu_0|Q_t]=\frac{\rho_1 \beta \cdot p(t, Q_t)+\rho_2 \alpha \cdot q(t,Q_t)}{\beta \cdot p(t,Q_t)+\alpha \cdot q(t, Q_t)},
\end{equation}
which concludes the proof.
\end{proof}

\begin{remark}
Notice that the optimal ExpMA strategy $\mathfrak{f}_2^*$ obtained in Theorem \ref{thm_MC_utility_C2} is \emph{semi}-explicit. 
To be precise, $\mathfrak{f}_2^*$ is indeed obtained explicitly once 
$u(t,x)$ and $v(t,x)$ in \eqref{def:uv} are identified for $t>0$ and $\frac{\rho_1}{\lambda}(1-e^{-\lambda t})\leq x \leq \frac{\rho_2}{\lambda}(1-e^{-\lambda t})$,
which is the purpose of the next proposition.
\end{remark}

\begin{proposition}\label{p1}
Let Assumptions \ref{assumption_MC} and \ref{assumption_MC_mu0} hold.
The functions $u(t,x)$ and $v(t,x)$, defined in \eqref{def:uv}, satisfy the following partial differential equation (PDE) system:
\begin{align}
\label{eqn_MC_ODE_system}
\begin{cases}
u_t+(\rho_1-\lambda x)u_x+\alpha u-\alpha v=0,  \\[1ex]
v_t+(\rho_2-\lambda x)v_x+\beta v-\beta u=0, 
\end{cases} 
\text{for } t>0 \text{ and } \frac{\rho_1}{\lambda}(1-e^{-\lambda t})\leq x \leq \frac{\rho_2}{\lambda}(1-e^{-\lambda t})
\end{align}
with boundary conditions
\begin{align}
u\left(t,\frac{\rho_1}{\lambda}(1-e^{-\lambda t})\right)&=e^{-\alpha t}, &
u\left(t,\frac{\rho_2}{\lambda}(1-e^{-\lambda t})\right)&=1, && t>0, \\
v\left(t,\frac{\rho_1}{\lambda}(1-e^{-\lambda t})\right)&=0, &
v\left(t,\frac{\rho_2}{\lambda}(1-e^{-\lambda t})\right)&=1, && t>0.
\end{align}

\end{proposition}

\begin{proof}
The proof is delayed to Appendix \ref{appen_p1}. 
\end{proof}

In the remaining part of this section, we study Problem \ref{prob_growth} for a subset of $\Cc_2$ strategies, denoted by $\tilde{\Cc}_2$,
\begin{align} \label{eqn_def_Cc2_tilde}
\tilde{\Cc}_2:=\{f \in \Cc_2: \; f(t,z) = \tilde{f}(z) \text{ for all } t \in [0,T]\}.
\end{align}
Namely, $\tilde{\Cc}_2$ includes all $\Cc_2$ strategies that are independent of time.
We shall explicitly obtain the optimal trading strategy and the long term growth rate in Theorem \ref{thm_MC_growth_C2}.

To begin our analysis, we note that as $t \to \infty$, we have
\begin{align}
Q_t\rightarrow Q_\infty:=\int_0^\infty e^{-\lambda s}\mu_s\, \dd s+\sigma\int_0^\infty e^{-\lambda s}\, \dd W_s-\frac{\sigma^2}{2\lambda}, \quad \text{a.s.}.
\end{align}
	Let $u_\infty$ and $v_\infty$ be the conditional c.d.f. of $\int_0^\infty e^{-\lambda s}\mu_s\, \dd s$ given $\mu_0=\rho_1$ and $\mu_0=\rho_2$, respectively. That is 
\begin{align} \label{eqn_def_u_v_infty}
u_\infty (x) := \Pb \left(\int_0^\infty e^{-\lambda s}\mu_s\, \dd s \le x \, \Big| \, \mu_0 = \rho_1 \right), \quad
v_\infty (x) := \Pb \left(\int_0^\infty e^{-\lambda s}\mu_s\, \dd s \le x \, \Big| \, \mu_0 = \rho_2 \right). 
\end{align}

The following lemma will be key.

\begin{lemma}
\label{lemma_limit_mu0}
Let Assumptions \ref{assumption_MC} and \ref{assumption_MC_mu0} hold, 
we have the following limit result:
\begin{align}\label{eq8}
\lim_{t \to \infty} \; \Eb[\mu_0|Q_t] = \Eb[\mu_0|Q_\infty]\quad \text{a.s.}.
\end{align}
\end{lemma} 

\begin{proof}
The proof is provided in Appendix \ref{appen_lemma_limit_mu0}.
\end{proof}

 We have the results below regarding the functions $u_\infty$ and $v_\infty$, defined in \eqref{eqn_def_u_v_infty}, and $\Eb[\mu_0|Q_\infty]$, which appears in Lemma \ref{lemma_limit_mu0}.

\begin{lemma}
\label{lem_limit_distribution}
Let Assumptions \ref{assumption_MC} and \ref{assumption_MC_mu0} hold, the functions $u_\infty$ and $v_\infty$, defined in \eqref{eqn_def_u_v_infty}, satisfy, for $\frac{\rho_1}{\lambda}<x<\frac{\rho_2}{\lambda}$, that
\begin{equation}
u_\infty(x)=c\int_{\frac{\rho_1}{\lambda}}^x (\rho_2-\lambda z) \cdot l(z)\, \dd z\quad\text{and}\quad v_\infty(x)=d\int_{\frac{\rho_1}{\lambda}}^x (\lambda z-\rho_1) \cdot l(z)\, \dd z,
\end{equation}
where
\begin{align}
l(z) &=(\lambda z-\rho_1)^{\frac{\alpha}{\lambda}-1} \cdot (\rho_2-\lambda z)^{\frac{\beta}{\lambda}-1}, \label{eqn_def_l}\\
c &=\dfrac{\lambda^2 \cdot \Gamma \(\frac{\alpha+\beta+\lambda}{\lambda} \)}{\beta(\rho_2-\rho_1)^{\frac{\alpha+\beta}{\lambda}} \cdot \Gamma\(\frac{\alpha}{\lambda}\)\cdot \Gamma\(\frac{\beta}{\lambda}\)}, \quad \text{and} \quad d=\frac{\beta c}{\alpha}, \label{eq14} \\
\text{with} \quad \Gamma(z) &= \int_0^\infty x^{z-1} \cdot e^{-x} \dd x.
\end{align}
That is, $u_\infty$ and $v_\infty$ are the c.d.f of (scaled and shifted) Beta distributions.
\end{lemma}

\begin{proof}
The proof is provided in Appendix  \ref{appen_lem_limit_distribution}.
\end{proof}

\begin{proposition}
\label{prop_limit_conditional_distribution}
Under Assumptions \ref{assumption_MC} and \ref{assumption_MC_mu0},
we have that
\begin{align}
\Eb[\mu_0|Q_\infty]=\dfrac{\lambda\displaystyle\int_{\frac{\rho_1}{\lambda}}^{\frac{\rho_2}{\lambda}}z \cdot l(z) \cdot \phi_\infty(Q_\infty-z)\, \dd z}{\displaystyle\int_{\frac{\rho_1}{\lambda}}^{\frac{\rho_2}{\lambda}} l(z) \cdot \phi_\infty(Q_\infty-z)\, \dd z},
\end{align}
where $l$ is given by \eqref{eqn_def_l} and $\phi_\infty$ is defined by
\begin{equation}
\label{eqn_def_phi_infinity}
\phi_\infty(x) :=\sqrt{\dfrac{\lambda}{\pi\sigma^2}}\cdot \exp \left( -\dfrac{\left[x+\frac{\sigma^2}{2\lambda}\right]^2}{\frac{\sigma^2}{\lambda}} \right), \quad x \in \Rb.
\end{equation}
\end{proposition}

\begin{proof}
	The conditional p.d.f. of $Q_\infty$ given $\mu_0=\rho_1$ satisfies
	\begin{equation}\label{eq13}
	p_\infty(x)=\int_{-\infty}^\infty u_\infty'(z) \cdot \phi_\infty(x-z)\,\dd z=c\int_{\frac{\rho_1}{\lambda}}^{\frac{\rho_2}{\lambda}}(\rho_2-\lambda z) \cdot l(z) \cdot \phi_\infty(x-z)\, \dd z.
	\end{equation}
	Similarly,
	\begin{align}
	q_\infty(x)=d\int_{\frac{\rho_1}{\lambda}}^{\frac{\rho_2}{\lambda}}(\lambda z-\rho_1) \cdot l(z) \cdot \phi_\infty(x-z)\, \dd z.
	\end{align}
	The above equations, together with \eqref{eq10}, \eqref{eq14} and Lemma \ref{lemma_limit_mu0},  imply the desired result.
\end{proof}

Recall that
\begin{align}
\kf_2^*(t,x)=\frac{1}{\sigma^2}\Eb[\mu_t|Z_t=x]=\frac{1}{\sigma^2}\Eb[\mu_0|Q_t=x].
\end{align}
Denote
\begin{align}
\label{eqn_def_g_inf}
\kg_\infty(x):=\frac{1}{\sigma^2}\Eb[\mu_0|Q_\infty=x]=\frac{\lambda\displaystyle\int_{\frac{\rho_1}{\lambda}}^{\frac{\rho_2}{\lambda}}z \cdot l(z) \cdot \phi_\infty(x-z)\, \dd z}{\sigma^2\displaystyle\int_{\frac{\rho_1}{\lambda}}^{\frac{\rho_2}{\lambda}}l(z) \cdot \phi_\infty(x-z)\, \dd z}.
\end{align}
We have the following theorem which provides solutions to Problem \ref{prob_growth} restricted to $\tilde{\Cc}_2$ strategies.
\begin{theorem}
\label{thm_MC_growth_C2}
Under Assumptions \ref{assumption_MC} and \ref{assumption_MC_mu0},
we have that
\begin{align}
\lim_{T\rightarrow\infty}\frac{1}{T}\Eb\left[\ln\left(\frac{\Pi_T^{\kf_2^*}}{\Pi_0}\right)\right]&=
\lim_{T\rightarrow\infty}\frac{1}{T}\Eb\left[\ln\left(\frac{\Pi_T^{\kg_\infty}}{\Pi_0}\right)\right]\\
&=\frac{c\beta\lambda^2(\rho_2-\rho_1)}{2\sigma^2(\alpha+\beta)}\int_{-\infty}^\infty\left[\frac{\left(\displaystyle\int_{\frac{\rho_1}{\lambda}}^{\frac{\rho_2}{\lambda}}z \cdot l(z) \cdot \phi_\infty(y-z)\,\dd z\right)^2}{\displaystyle\int_{\frac{\rho_1}{\lambda}}^{\frac{\rho_2}{\lambda}} l(z) \cdot \phi_\infty(y-z)\, \dd z}\right]\, \dd y,
\end{align}
where  $l$, constant $c$, and $\phi_\infty$ are defined in \eqref{eqn_def_l}, \eqref{eq14}, and \eqref{eqn_def_phi_infinity} respectively.

In particular, this implies that $\kg_\infty(\cdot)$, given by \eqref{eqn_def_g_inf}, is an optimal ExpMA strategy to Problem \ref{prob_growth} within the $\tilde{\Cc}_2$ class, where $\tilde{\Cc}_2$ is defined in \eqref{eqn_def_Cc2_tilde}.
\end{theorem}

\begin{proof}
We have that
\begin{align}
\lim_{T\rightarrow\infty}\frac{1}{T}\Eb\left[\ln\left(\frac{\Pi_T^{\kf_2^*}}{\Pi_0}\right)\right]&=\lim_{T\rightarrow\infty}\dfrac{1}{2\sigma^2T}\int_0^T\Eb\left[(\Eb[\mu_t|Z_t])^2\right]\, \dd t \\
\label{e14}&=\frac{1}{2\sigma^2}\Eb\left[ \left(\Eb[\mu_0|Q_\infty] \right)^2\right]\\
&=\frac{1}{2\sigma^2}\int_{-\infty}^\infty\left(\kg_\infty(y)\right)^2\cdot\left(\frac{\beta}{\alpha+\beta}p_\infty(y)+\frac{\alpha}{\alpha+\beta}q_\infty(y)\right)\, \dd y\\
&=\frac{c\beta\lambda^2(\rho_2-\rho_1)}{2\sigma^2(\alpha+\beta)}\int_{-\infty}^\infty\left[\frac{\left(\displaystyle\int_{\frac{\rho_1}{\lambda}}^{\frac{\rho_2}{\lambda}}z \cdot l(z) \cdot \phi_\infty(y-z)\,\dd z\right)^2}{\displaystyle\int_{\frac{\rho_1}{\lambda}}^{\frac{\rho_2}{\lambda}} l(z) \cdot \phi_\infty(y-z)\, \dd z}\right]\, \dd y,
\label{e15}
\end{align}
where the second equality comes from the reversibility condition.

It is easy to show that $\kg_\infty$ is a bounded and continuous function. Thus, we obtain
\begin{align}
\lim_{T\rightarrow\infty}\frac{1}{T}\Eb\left[\ln\left(\frac{\Pi_T^{\kg_\infty}}{\Pi_0}\right)\right]&=\lim_{T\rightarrow\infty}\frac{1}{T}\int_0^T\Eb\left[\mu_t \cdot \kg_\infty(Z_t)-\frac{1}{2}\sigma^2 \cdot \kg_\infty^2(Z_t)\right]\, \dd t \\
&=\lim_{T\rightarrow\infty}\frac{1}{T}\int_0^T\Eb\left[\mu_0 \cdot \kg_\infty(Q_t)-\frac{1}{2}\sigma^2 \cdot \kg_\infty^2(Q_t)\right]\, \dd t\\
&=\Eb\left[\mu_0 \cdot \kg_\infty(Q_\infty)-\frac{1}{2}\sigma^2 \cdot \kg_\infty^2(Q_\infty)\right]\\
&=\frac{1}{2\sigma^2}\Eb\left[ \big(\Eb[\mu_0|Q_\infty] \big)^2\right].
\end{align}
This, together with \eqref{e14} and \eqref{e15}, implies the result.
\end{proof}

\section{Monte Carlo Investigation}
\label{sec:monte.carlo}

Although our theoretical results justify the use of trading strategies for optimizing expected utility in the case of an OU or a two-state Markov chain drift, there are two shortcomings of our analysis from a practical point of view.  First, we have not examined the most widely-used measure of portfolio performance -- the Sharpe ratio.  Second, our modeling framework does not take into account transaction costs.  In this section, we address these two shortcomings and conduct sensitivity analysis using Monte Carlo simulations.
 
\begin{table}[h]
	\centering
	\begin{tabular}{|c|c|c|c|} \hline
		$\sig$ & $\kappa$ & $\bar{\mu}$ & $\delta$ \\ \hline
		0.0436 & 0.0226 & 0.0034 & 8.2404e-04 \\ \hline
	\end{tabular}
	\caption{Model Parameters}
	\label{table_parameters}
\end{table}

In the numerical studies, we assume the drift is modeled by an OU process and Assumptions \ref{assumption_OU}-\ref{assumption_normal_mu0} hold true. 
The base model parameters (in monthly units) are chosen as shown in Table \ref{table_parameters}, which are modified from \cite{wachter2002portfolio}.
We assume there are 21 trading days per month (equivalent to 252 trading days per year), and we choose a one-day time step in our Monte Carlo discretization, i.e., $\Delta t = \frac{1}{21}$. 
The finite time horizon $T$ in Problem \ref{prob_utility} is the number of months of investment, e.g., $T=12$ means an investment period of one year. Given $T$,  the number of trading days is then $21 \times T$. According to \eqref{eqn_X}, \eqref{eqn_dY} and \eqref{eqn_OU_drift}, we discretize the log price process $X$, the ExpMA process $Y$ and the drift process $\mu$ by 
\begin{align}
X_{i+1} &= X_i + \left(\mu_i - \frac{1}{2} \sig^2 \right) \Delta t + \sig \sqrt{\Delta t} z_i, \\
Y_{i+1} &= Y_i + \lam ( X_i - Y_i ) \Delta t, \label{eqn_Y_update} \\
\mu_{i+1} &= \mu_i + \kappa (\bar{\mu} - \mu_i) \Delta t + \delta \sqrt{\Delta t} \bar{z}_i,
\end{align}
where $(z_i)$ and $(\bar{z}_i)$ are independent random variables sampled from a standard normal distribution. 
{In practice, to compute the ExpMA, a time period $\mathcal{P}$ (in number of days) is specified, and the weight in the most recent price -- corresponding to $\lambda \Delta t$ in equation \eqref{eqn_Y_update} for $Y$ above -- is given by $\frac{2}{\mathcal{P} + 1}$.
Common choices for $\mathcal{P}$ include 10, 20, 50, 100, and 200, which converts to $\lambda = \{\frac{42}{11}, 2, \frac{42}{51}, \frac{42}{101}, \frac{42}{201} \}$. 
We choose $\lambda=2$ as the base parameter for the numerical analysis that follows.
}



\subsection{Performance Analysis}

We select $T=24$
and $\lambda = 2$. We consider three optimal ExpMA strategies: 
(i) the optimal $\Cc_1$ ExpMA strategy for utility maximization Problem \ref{prob_utility} (see Theorem \ref{thm_OU_utility_linear}),
(ii) the optimal $\Cc_2$ ExpMA strategy for utility maximization Problem \ref{prob_utility} (see Theorem \ref{thm_OU_utility_general}), and 
(iii) the optimal ExpMA strategy for growth maximization Problem \ref{prob_growth} (see Theorem \ref{thm_OU_growth}). 
For comparison, we also consider a buy-and-hold (BH) strategy.  
We set the initial state of all four strategies to be the same, beginning with one share of the risky asset and zero value in the risk-free asset.
We run 10,000 simulations and summarize the results in Table \ref{table_performance}, where
``Return'' is the simple return of a strategy over the entire investment period, ``Avg. Daily Return'' is the average daily return of a strategy, and the ``Sharpe ratio'' is computed using the daily simple return of a strategy.

%

\begin{table}[h]
\centering
\begin{tabular}{|c|c|c|c|c|} \hline
	Strategy & Utility-$\Cc_1$ & Utility-$\Cc_2$ & Growth & BH \\ \hline
	Return & 17.3730\% & 17.3731\%  & 17.3764\% & 8.8816\% \\ \hline
	Avg. Daily Return & 0.0290\%& 0.0290\%& 0.0290\%& 0.0161\%\\ \hline
	Sharpe Ratio & 0.0161 & 0.0161 & 0.0161 & 0.0169 \\ \hline
\end{tabular}
\caption{Performance of Strategies when $T=20$ and $\lam = 2$}
\label{table_performance}
\end{table}

We notice that all three optimal ExpMA strategies perform nearly the same. 
Indeed, the strategies are quite similar.
Under the given parameters, we have
\begin{align}
\mathsf{f}_1^*(z) &= a_1^* \cdot z + b_1^* = 8.1147 \cdot z + 1.7788,  &&(\text{optimal utility-$\Cc_1$ ExpMA Strategy})\\
\mathsf{f}_\infty^*(z) &= a_\infty \cdot z + b_\infty = 8.1580 \cdot z + 1.7786. 
&&(\text{optimal growth ExpMA Strategy})
\end{align}
In addition,  in the optimal utility-$\Cc_2$ ExpMA Strategy, 
\begin{align}
\mathsf{f}_2^*(t,z) &= a_2^*(t) \cdot z + b_2^*(t),
\end{align} 
we have $a_2^*(t) \to a_\infty$ after 142 days and $b_2^*(t) \to b_\infty$ after 59 days.

We also observe that the optimal ExpMA strategies deliver excellent average returns, nearly doubling the return of the buy-and-hold strategy. However, the buy-and-hold strategy achieves a slightly higher Sharpe ratio. 
{Thus, if the Sharpe ratio, rather than expected utility or long-run growth, is the primary measure of investment performance, it may be best for an investor to employ a buy-and-hold strategy.}



\subsection{Sensitivity Analysis}
In this section, we conduct a sensitivity analysis, designed to examine the impact of various factors on the optimal ExpMA strategies.
Since the performance of all three optimal ExpMA strategies is very close in both return and Sharpe ratio, we only consider the optimal growth ExpMA strategy in what follows. 
We still use the parameters in the previous subsection, and only allow one parameter to vary in each study.

We first examine the role of ExpMA parameter $\lam$ in investment performance, {which is a key parameter in the definition of the ExpMA; see \eqref{eqn_def_ExpMA}.} 
In addition to $\lam = 2$ in Table \ref{table_performance}, we also include {$\lam = \frac{42}{11}, \frac{42}{51}, \frac{42}{101}, \frac{42}{201}$} 
Based on the results of our simulations, optimal ExpMA strategies that use a {smaller $\lam$ ($\lam =  \frac{42}{101}, \frac{42}{201}$)} provide higher returns than those using a {larger $\lam$ ($\lam = \frac{42}{11}, 2, \frac{42}{51}$)}, but deliver a poorer Sharpe ratio. 
{In general, the ``ideal'' $\lam$ will depend on a trader's measure of performance as well as the dynamics of $\mu$.}

\begin{table}[h]
\centering
\begin{tabular}{|c|c|c|c|c|c|} \hline
ExpMA Parameter {$\lam$} &  {$\frac{42}{11}$} &   {2} & {$\frac{42}{51}$} & {$\frac{42}{101}$} & {$\frac{42}{201}$}  \\ \hline
Return & 16.6340\% & 17.3764\% & 17.6406\% & 18.9976\% & 18.6209\%  \\ \hline
Avg. Daily Return & 0.0281\%& 0.0290\%& 0.0293\%& 0.0308\% & 0.0296\%\\ \hline
Sharpe Ratio & 0.0160 & 0.0161 & 0.0149 & 0.0138 & 0.0106 \\ \hline
\end{tabular}
\caption{Impact of Moving Average Window on Performance}
\label{table_window}
\end{table}

Next we investigate the effect of the time horizon $T$ effect on the performance of the optimal ExpMA strategies. 
In addition to $T=24$, we examine horizons of $T=12$ (1 year), $T=60$ (5 years), $T=120$ (10 years) and $T=360$ (30 years) in the comparison. Since $T$ is the factor under consideration, we do not report the simple return over the entire investment period $T$. Instead, only average daily return and Sharpe ratio are reported in Table \ref{table_time}. 
We observe that the average daily return is not sensitive to the change of $T$, but Sharpe ratio does improve as $T$ increases. Such an observation is consistent with the theoretical conclusion that the optimal ExpMA strategy is the solution to the long-run ($T \to \infty$)) growth maximization problem.

\begin{table}[h]
	\centering
	\begin{tabular}{|c|c|c|c|c|c|} \hline
		Time $T$ & 12 & 24 & 60 & 120 & 360 \\ \hline
		Avg. Daily Return & 0.0289\%& 0.0290\%& 0.0286\%& 0.0293\% & 0.0292\%\\ \hline
		Sharpe Ratio & 0.0153 & 0.0161 & 0.0163 & 0.0169 & 0.0170 \\ \hline
	\end{tabular}
	\caption{Impact of Time Horizon on Performance}
	\label{table_time}
\end{table}

We end this subsection by analyzing how stock volatility $\sig$ affects the performance of the optimal ExpMA strategies. In addition to the base value of $\sig = 0.0436$, we also consider $\sig = 0.0523$ (20\% increase) and $\sig = 0.0349$ (20\% decrease) in the analysis. 
The results in Table \ref{table_vol} clearly show that the performance of both optimal ExpMA and buy-and-hold strategies is negatively correlated with the stock volatility $\sig$, i.e., both strategies perform well (resp. poor) when $\sig$ is small (resp. big). 
Since volatility is bigger in a bear market, the above conclusion indicates that using ExpMA strategies in a bear market may not be ideal.

\begin{table}[h]
	\centering
	\begin{tabular}{|c|c|c|c|c|c|c|} \hline
		Volatility $\sig$ & \multicolumn{2}{|c|}{$\sig=0.0349$} & \multicolumn{2}{|c|}{$\sig=0.0436$} & \multicolumn{2}{|c|}{$\sig=0.0523$} \\ \hline
		Strategy & ExpMA & BH & ExpMA & BH & ExpMA & BH  \\ \hline
		Return & 29.1054\% & 8.6810\% & 17.3764\% & 8.8816 & 11.3129\% & 8.8204\%  \\ \hline
		Avg. Daily Return & 0.0446\%& 0.0158\%& 0.0290\% & 0.0161\% & 0.0199\% & 0.0160\%\\ \hline
		Sharpe Ratio & 0.0197 & 0.0208 & 0.0161 & 0.0169 & 0.0133 & 0.0140 \\ \hline
	\end{tabular}
	\caption{Impact of Stock Volatility on Performance}
	\label{table_vol}
\end{table}

\subsection{Transaction Costs}
Our theoretical results are derived under an assumption of zero transaction costs.  
By comparison, \cite{dai2010trend, dai2016optimal} take into account transaction costs when studying optimal buying/selling times. 
In order to see how transactions costs affect the performance of trading strategies based on ExpMAs, we follow \cite{dai2010trend} and suppose that trading the risk-free asset is frictionless, but trading the risky asset is subject to proportional costs of $\omega$. The stock price process $S$ models the mid-price, hence the cost of purchasing one share at time $t$ is $(1+\omega) S_t$ and the revenue of selling one share at time $t$ is $ (1-\omega) S_t$. In \cite{dai2010trend}, the authors choose $\omega = 0.1\%$.  Here, we consider $\omega= \{0.1\%, \, 0.5\%, \, 1\%\}$. 
In what follows, we focus only on the optimal growth ExpMA Strategy; see \eqref{eqn_def_f_infty}.

Let us explain how portfolio wealth is updated from day $i$ to day $i+1$ when proportional transaction costs are taken into account.
\begin{enumerate}
	\item At day $i+1$ before rebalancing, the optimal ExpMA portfolio wealth $\Pi_{(i+1)-}$ is given by
	\begin{align}
\Pi_{(i+1)-} = (1 - f_i) \cdot \Pi_i + \dfrac{f_i \Pi_i}{e^{X_i}} \cdot e^{X_{i+1}},
	\end{align}
	where $f_i$ is the proportion {of wealth invested in the risky asset during $[i, i+1)$} and $\frac{f_i \Pi_i}{e^{X_i}}$ is the number of shares invested in the risky asset during $[i, i+1)$.
	
	\item {At the time of} rebalancing, from \eqref{eqn_def_f_infty}, the new investment weight $f_{i+1}$ is given by
	\begin{align}
	f_{i+1} = a_\infty \cdot (X_{i+1} - Y_{i+1}) + b_\infty.
	\end{align}
	Suppose the change of shares in the risky asset is $\Delta_{i+1}$.  Then the number of shares in the risky asset after rebalancing is $f_i\Pi_i / e^{X_i} + \Delta_{i+1}$. 
	
	\item Denoting by $\Pi_{i+1}$ the wealth after rebalancing, we have
	\begin{align}
	\Pi_{i+1} &= \Pi_{(i+1)-} - \omega |\Delta_{i+1}| e^{X_{i+1}}, 
	\label{eqn_update_wealth}\\
	 f_{i+1} \cdot \Pi_{i+1} &= \left(f_i\Pi_i / e^{X_i} + \Delta_{i+1} \right) \cdot  e^{X_{i+1}} .
	\end{align}
	Solving the above equations for $\Del_{i+1}$ yields 
	\begin{align}
	\label{eqn_share_change}
	\Delta_{i+1} = \begin{cases}
	\frac{ f_{i+1}  \Pi_{(i+1)-} - f_i \Pi_i \exp(X_{i+1}-X_i) } {\left(1 + \omega \cdot f_{i+1}\right) \exp(X_{i+1})}, & \text{ if } f_{i+1} \ge f_{i} \\[2ex]
	\frac{ f_{i+1}  \Pi_{(i+1)-} - f_i \Pi_i \exp(X_{i+1}-X_i) } {\left(1 - \omega \cdot f_{i+1}\right) \exp(X_{i+1})}, & \text{ if } f_{i+1} < f_{i} 
	\end{cases}.
	\end{align}

\end{enumerate}
 
\begin{table}[h]
	\centering
	\begin{tabular}{|c|c|c|c|c|c|} \hline
		Transaction Cost $\omega$ &  0.1\% & 0.5\% & 1\% & 0\% & BH \\ \hline
		Return & 13.4712\% & -1.4013\% & -17.2807\% & 17.3764\% & 8.8816\%  \\ \hline
		Avg. Daily Return & 0.0219\% & -0.0405\%& 0.0286\%& 0.0290\% & 0.0161\%\\ \hline
		Sharpe Ratio & 0.0119 & -0.0044 & -0.0247 & 0.0161 & 0.0169 \\ \hline
	\end{tabular}
	\caption{Impact of Transaction Costs on Performance}
	\label{table_cost}
\end{table}

We use \eqref{eqn_update_wealth} and \eqref{eqn_share_change} to update the wealth at day $(i+1)$ after rebalancing. Our numerical findings are included in Table \ref{table_cost}, where the last two columns are repeated from Table \ref{table_performance} for the purposes of comparison. 
As we can see from Table \ref{table_cost}, when transaction costs are small, the optimal ExpMA strategy still performs reasonably well. However, as transaction costs increase, the optimal ExpMA strategy is no longer profitable and use of this strategy is no longer advised. 
{Hence, for markets with large transaction costs, optimal ExpMA strategies may not be appropriate. }

\section{Conclusion}
\label{sec_conclusion}
Moving averages are widely used indicators in technical analysis and are commonly applied by practitioners to construct trading strategies. In this paper, we provide a mathematical analysis of trading strategies that are constructed using the ExpMA of the risky asset. Namely, we study the classical optimal investment problems for ExpMA strategies. The drift process of the risky asset in our framework is modeled by either an OU process or a CTMC. We obtain the optimal ExpMA strategy in explicit forms for the affine class ($\Cc_1$) and the square-integrable class ($\Cc_2$) under two optimization criteria: logarithmic utility maximization and long term growth rate maximization.  
We find that, in the case of an OU drift, the optimal ExpMA strategy to the logarithmic utility maximization problem under the $\Cc_2$ class is in affine form, and the solution to the long term growth rate maximization problem is the same for both the $\Cc_1$ and $\Cc_2$ classes. 
By comparison, in the case of a CTMC drift, the optimal strategies of the two maximization problems are significantly different under the $\Cc_1$ and $\Cc_2$ classes. 
In general, our numerical results show that optimal ExpMA strategies deliver excellent returns in comparison to a buy-and-hold strategy.  However, the buy-and-hold strategy has a slightly higher Sharpe ratio. 
When transaction costs are large, our studies suggest caution when using ExpMAs in trading. 

\subsection*{Acknowledgments}
The authors would like to express their gratitude to Douglas Service, who first brought the problems studied in this paper to the authors' attention. 
We are grateful to an associate editor and two anonymous referees, whose comments have helped improve an earlier version of the paper. Bin Zou acknowledges a start-up grant from the University of Connecticut.

\appendix

\section{Appendix}
\label{sec_appendix}

\subsection{Proof of Theorem \ref{thm_OU_utility_linear}}
\label{sec:proof-1}
\begin{proof}
	We begin the proof by solving the SDE \eqref{eqn_OU_drift} for $\mu$, which yields
	\begin{align} \label{eqn_OU_drift_expression}
	\mu_t = \bar{\mu} \left( 1 - e^{-\kappa t} \right) + e^{-\kappa t} \mu_0 + \delta e^{-\kappa t} N_1(t), \text{ where } 	N_1(t) := \int_0^t e^{\kappa s} \dd \bar{W}_s \sim \Nc \left(0, \, \frac{e^{2\kappa t} -1}{2\kappa}  \right).
	\end{align}
	With $\mu$ given by \eqref{eqn_OU_drift_expression}, the mean and variance of $\mu_t$ are obtained by
	\begin{align}
	m_1(t) &:=\Eb [\mu_t] =  \bar{\mu}  + \left( m_1(0) - \bar{\mu} \right) e^{-\kappa t} 
	= \mb_1^1 + \mb_2^1 \cdot e^{-\kappa t} ,
	\label{eqn_drift_mean}\\
	v_1(t) &:=\Vb [\mu_t] = \frac{\delta^2}{2 \kappa} +  \left(v_1(0) - \frac{\delta^2}{2 \kappa} \right) e^{-2\kappa t}
	=\vb_1^1 + \vb_2^1 \cdot e^{-2\kappa t}, \label{eqn_drift_variance}
	\end{align}
	where $\mb_1^1 := \bar{\mu}$,  $\mb_2^1:=  m_1(0) - \bar{\mu}$, 
	$\vb_1^1 :=\frac{\delta^2}{2 \kappa}$, and $\vb_2^1 := v_1(0) - \frac{\delta^2}{2 \kappa}$.
	
	Solving the SDE of $Z$ gives
	\begin{equation} \label{eqn_Z_first_expression}
	Z_t = \int_0^t \left(\mu_s - \frac{1}{2} \sigma^2 \right) e^{-\lambda (t-s)} \dd s + \sigma e^{-\lambda t} N_2(t),\text{ where } 
	N_2(t) := \int_0^t e^{\lambda s} \dd W_s \sim \Nc \left(0, \, \frac{e^{2\lambda t} -1}{2\lambda} \right).
	\end{equation}
	The first term of $Z_t$ in \eqref{eqn_Z_first_expression} can be rewritten as
	\begin{align}
	\int_0^t \left(\mu_s - \frac{1}{2} \sigma^2 \right) e^{-\lambda (t-s)} \dd s = e^{-\lambda t} \int_0^t \mu_s  e^{\lambda s} \dd s  - \frac{\sigma^2}{2\lambda} \left(1- e^{-\lambda t}  \right).
	\end{align}
	By applying integration by parts to the first integral, we obtain
	\begin{align}
	\int_0^t \mu_s  e^{\lambda s} \dd s = \frac{\kappa \bar{\mu}}{\lambda (\kappa - \lambda)} \left( e^{\lambda t}  - 1 \right) + \frac{\mu_0}{\kappa - \lambda}  - \frac{e^{\lambda t} \mu_t}{\kappa - \lambda}   + \frac{\delta}{\kappa - \lambda} N_3(t), \text{ where } N_3(t):=\int_0^t e^{\lambda s} \dd \bar{W}_s.
	\end{align}
	Notice that $N_3(t)$ and $N_2(t)$ have the same distribution, and are independent.
	Finally, we rewrite $Z_t$ by
	\begin{align} \label{eqn_Z_second_expression}
	Z_t = \frac{1}{\lambda} \left(\frac{\kappa \bar{\mu}}{\kappa - \lambda} - \frac{\sigma^2}{2}\right) \left( 1- e^{-\lambda t}  \right) + \frac{1}{\kappa - \lambda} e^{-\lambda t} \mu_0 - \frac{1}{\kappa - \lambda} \mu_t + \frac{\delta}{\kappa - \lambda} e^{-\lambda t} N_3(t)+ \sigma e^{-\lambda t} N_2(t).
	\end{align}
	In establishing \eqref{eqn_Z_second_expression}, we need the technical condition $\kappa \neq \lambda$ imposed in Assumption \ref{assumption_OU}.

	We find the covariances of the random variables that appear in \eqref{eqn_OU_drift_expression}and \eqref{eqn_Z_second_expression} as
	\begin{align}
	\text{Co}\Vb (\mu_0, N_1(t)) &= \text{Co}\Vb (\mu_0, N_2(t)) = \text{Co}\Vb (\mu_0, N_3(t)) =
	\text{Co}\Vb (N_1(t), N_2(t)) =  \text{Co}\Vb (N_2(t), N_3(t))=0,\\
	\text{Co}\Vb (N_1(t), N_3(t)) &= \Eb [N_1(t) N_3(t)]= \frac{1}{\kappa + \lambda} \left( e^{(\kappa + \lambda) t} - 1 \right), \\
	\text{Co}\Vb(\mu_0,\mu_t) &= e^{-\kappa t} v_1(0), \quad
	\text{Co}\Vb(\mu_t, N_3(t)) = \frac{\delta}{\kappa + \lambda} e^{-\kappa t}
	\left( e^{(\kappa + \lambda) t} - 1 \right).
	\end{align}
	
	Now we are ready to find the mean and the variance of $Z_t$:
	\begin{align}
	m_2(t) := \Eb [Z_t] 
	&=\frac{2\bar{\mu}-\sigma^2}{2\lambda} + \left[\frac{\lambda m_1(0) - \kappa \bar{\mu}}{\lambda(\kappa-\lambda)} + \frac{\sigma^2}{2 \lambda}\right]e^{-\lambda t}+ \frac{\bar{\mu} - m_1(0)}{\kappa - \lambda} e^{-\kappa t}\\
	&= \mb_1^2 + \mb_2^2 \cdot e^{-\lambda t} + \mb_3^2 \cdot e^{-\kappa t}, \label{eqn_OU_m2}
	\\
	v_2(t) :=\Vb [Z_t] 
	&= \frac{\sigma^2}{2\lambda} + \frac{\delta^2}{2 \kappa \lambda (\kappa + \lambda)}
	+ \left[\frac{1}{(\kappa - \lambda)^2} \left(v_1(0) - \frac{\delta^2}{2\lambda}\right) - \frac{\sigma^2}{2\lambda} \right] e^{-2 \lambda t} \\
	&\quad + \frac{1}{(\kappa - \lambda)^2} \left(v_1(0) - \frac{\delta^2}{2\kappa}\right) e^{-2 \kappa t}
	- \frac{2}{(\kappa - \lambda)^2}\left(v_1(0) - \frac{\delta^2}{\kappa + \lambda}\right) e^{-( \kappa + \lambda) t} \\
	&= \vb_1^2 + \vb_2^2 \cdot e^{-2\lambda t} + \vb_3^2 \cdot e^{-2\kappa t}
	+ \vb_4^2 \cdot e^{- (\kappa + \lambda) t}, \label{eqn_OU_v2}
	\end{align}
	where we have used the definitions $\mb_1^2:=\dfrac{2\bar{\mu}-\sigma^2}{2\lambda}$,  
	$\mb_2^2:=\dfrac{\lambda m_1(0) - \kappa \bar{\mu}}{\lambda(\kappa-\lambda)} + \dfrac{\sigma^2}{2 \lambda}$, 
	$\mb_3^2:=\dfrac{\bar{\mu} - m_1(0)}{\kappa - \lambda}$,
	$\vb_1^2:=\dfrac{\sigma^2}{2\lambda} + \dfrac{\delta^2}{2 \kappa \lambda (\kappa + \lambda)}$, 
	$\vb_2^2:= \dfrac{1}{(\kappa - \lambda)^2} \left(v_1(0) - \dfrac{\delta^2}{2\lambda}\right) - \dfrac{\sigma^2}{2\lambda}$,
	$\vb_3^2:=\dfrac{1}{(\kappa - \lambda)^2} \left(v_1(0) - \dfrac{\delta^2}{2\kappa}\right)$, 
	$\vb_4^2:=- \dfrac{2}{(\kappa - \lambda)^2}\left(v_1(0) - \dfrac{\delta^2}{\kappa + \lambda}\right)$.
	
	Similarly, we obtain $\mathbb{E}[\mu_t Z_t]$ by
	\begin{align}
	m_3(t):= \Eb [\mu_t Z_t] &=\frac{1}{\lambda} \left(\frac{\kappa \bar{\mu}}{\kappa - \lambda} - \frac{\sigma^2}{2}\right) \left( 1- e^{-\lambda t}  \right) m_1(t)+\frac{e^{-\lambda t}}{\kappa-\lambda}\left[\bar\mu(1-e^{-\kappa t})m_1(0)+e^{-\kappa t}(m_1^2(0)+v_1(0))\right] \\
	&\quad -\frac{1}{\kappa-\lambda}\left[v_1(t)+m_1^2(t)\right]+\frac{\delta^2}{\kappa^2-\lambda^2}\left(1-e^{-(\kappa+\lambda)t}\right) \\
	&= \mb_1^3 + \mb_2^3 \cdot e^{-2\kappa t} + \mb_3^3 \cdot e^{- (\kappa+\lambda)t}
	+ \mb_4^3 \cdot e^{-\kappa t} + \mb_5^3 \cdot e^{-\lambda t}, \label{eqn_OU_m3}
	\end{align}
	where 
	$\mb_1^3:= \bar{\mu} \mb_1^2 + \dfrac{\delta^2}{2\kappa(\kappa + \lambda)}$, 
	$\mb_2^3:= - \bar{\mu} \mb_3^2 - \dfrac{v_1(0)}{\kappa - \lambda}  -
	\dfrac{m_1(0)}{\kappa - \lambda}  \mb_2^1 + \dfrac{\delta^2}{2\kappa(\kappa - \lambda)}$,
	$\mb_3^3 :=  -  \left(\dfrac{\kappa \bar{\mu}}{\lambda(\kappa - \lambda)} - \dfrac{\sigma^2}{2 \lambda}\right) m_1(0)
	+ \dfrac{m_1(0)^2 + v_1(0) }{\kappa - \lambda} -\dfrac{\delta^2}{\kappa^2 - \lambda^2} - \bar{\mu} \mb_2^2$,
	$\mb_4^3:=- \bar{\mu} \mb_1^2 +  \bar{\mu} \mb_3^2 +  \left(\dfrac{\kappa \bar{\mu}}{\lambda(\kappa - \lambda)} - \dfrac{\sigma^2}{2 \lambda}\right) m_1(0)
	- \dfrac{\bar{\mu} m_1(0)}{\kappa - \lambda} $, and
	$\mb_5^3:= \bar{\mu} \mb_2^2$.
	
	Finally, we are able to compute $A(T)$, $B(T)$, $C(T)$, and $D(T)$ as follows:
	\begin{align}
	A(T)&=\int_0^T \Eb[\mu_t Z_t] \dd t &\hspace{-2em} &=\mb_1^3 T + \frac{1}{2 \kappa} \mb_2^3 \left(1 - e^{-2\kappa T} \right) + \frac{1}{\kappa+\lambda} \mb_3^3\left(1 - e^{-(\kappa+\lambda) T} \right)\\
	&\; &\hspace{-2em} &\quad + \frac{1}{\kappa} \mb_4^3 \left(1 - e^{-\kappa T} \right)
	+ \frac{1}{\lambda} \mb_5^3\left(1 - e^{-\lambda T} \right), \label{eqn_OU_A}\\
	B(T)&=\int_0^T\Eb[\mu_t] \dd t &\hspace{-2em} &=\mb_1^1 T + \frac{1}{\kappa}\mb_2^1 \left(1 - e^{-\kappa T} \right), \label{eqn_OU_B}\\
	C(T)&=\int_0^T \Eb[Z_t^2] \dd t &\hspace{-2em} &=\left( (\mb_1^2)^2 + \vb_1^2\right)T
	+ \frac{1}{2\lambda} \left( (\mb_2^2)^2 + \vb_2^2\right)\left(1 - e^{-2\lambda T} \right) \\
	&\; &\hspace{-2em} &\quad + \frac{1}{2\kappa} \left( (\mb_3^2)^2 + \vb_3^2\right)\left(1 - e^{-2\kappa T} \right)+ \frac{2\mb_2^2 \mb_3^3 + \vb_4^2}{\kappa+\lambda} \left(1 - e^{-(\kappa+\lambda) T} \right)\\
	&\; &\hspace{-2em} &\quad +\frac{2}{\lambda} \mb_1^2 \mb_2^2 \left(1 - e^{-\lambda T} \right)
	+\frac{2}{\kappa} \mb_1^2 \mb_2^2 \left(1 - e^{-\kappa T} \right), \label{eqn_OU_C}\\
	D(T)&=\int_0^T\Eb[Z_t] \dd t &\hspace{-2em} & =\mb_1^2 T +\frac{1}{\lambda}\mb_2^2 \left(1 - e^{-\lambda T} \right)+ \frac{1}{\kappa}\mb_3^2 \left(1 - e^{-\kappa T} \right). \label{eqn_OU_D}
	\end{align}
	
	Notice that all the expressions in \eqref{eqn_OU_A}, \eqref{eqn_OU_B}, \eqref{eqn_OU_C}, and \eqref{eqn_OU_D} are fully \emph{explicit}, and only depend on the model parameters from Assumptions \ref{assumption_OU} and \ref{assumption_normal_mu0}.
\end{proof}

\subsection{Proof of Theorem \ref{thm_MC_utility_C1}}
\label{appen_thm_MC_utility_C1}
\begin{proof}
By solving the SDE of $Z_t$, we obtain \eqref{eqn_n2t} through
\begin{align}
n_2(t) &=\Eb[Z_t]
=\frac{1}{\lambda} \left(n_1-\frac{1}{2}\sigma^2\right) \left(1-e^{-\lambda t}\right).
\end{align}
Further, we have
\begin{align}
n_3(t)&=\Eb[\mu_tZ_t]=e^{-\lambda t}\int_0^t\Eb[\mu_s\mu_t]e^{\lambda s}\, \dd s-\frac{\sigma^2}{2\lambda}(1-e^{-\lambda t}) n_1 +\sigma e^{-\lambda t}\Eb[N_2(t)\mu_t], 
\label{eqn_n3}
\end{align}
where $n_1(t)$ is computed in \eqref{eqn_n1}.

Note that, for $s\leq t$, we have
\begin{align}
\Eb[\mu_s\mu_t] &=\Eb[\mu_s\Eb[\mu_t|\mu_s]]
=n_1^2+ \gamma e^{-(\alpha+\beta)(t-s)}, 
\end{align}
where 
\begin{align}
\gamma:=\mathbb{V}[\mu_t]=\frac{\alpha\beta}{(\alpha+\beta)^2}(\rho_1-\rho_2)^2.
\end{align}
This, together with \eqref{eqn_n3}, implies \eqref{eqn_n3t}.

We next compute
\begin{align}Z_t^2&=e^{-2\lambda t}\int_0^t\int_0^t\mu_s\mu_ve^{\lambda s}e^{\lambda v}\, \dd s\,\dd v+\frac{\sigma^4}{4\lambda^2}\left(1-e^{-\lambda t}\right)^2+\sigma^2e^{-2\lambda t}N_2^2(t)\\
&\quad -\frac{\sigma^2}{\lambda}e^{-\lambda t}\left(1-e^{-\lambda t}\right)\int_0^t\mu_se^{\lambda s}\, \dd s+2\sigma e^{-2\lambda t}N_2(t)\int_0^t\mu_se^{\lambda s}\, \dd s-\frac{\sigma^3}{\lambda}\left(1-e^{-\lambda t}\right)e^{-\lambda t}N_2(t).
\end{align}
Hence, \eqref{eqn_n4t} is shown.
\end{proof}

\subsection{Proof of Proposition \ref{p1}}
\label{appen_p1}
\begin{proof}
	For all $0 < h \ll t$, denote by $\#(h)$ the number of jumps for the drift $\mu$ in $(0,h]$. We have that 
	\begin{align}
	\Pb \left( Q_{1,t} \leq x\,\Big|\,\mu_0=\rho_1\right)
	=\quad &\Pb\left(\int_0^te^{-\lambda s}\mu_s\, \dd s\leq x\,\Big|\,\mu_0=\rho_1, \quad 
	\#(h) = 0\right)\cdot \Pb(\#(h) = 0\,|\,\mu_0=\rho_1)\\
	+\; &\Pb\left(\int_0^te^{-\lambda s}\mu_s\, \dd s\leq x\,\Big|\,\mu_0=\rho_1, \quad 
	\#(h) = 1\right)\cdot \Pb(\#(h) = 1\,|\,\mu_0=\rho_1)\\
	+\; &\Pb\left(\int_0^te^{-\lambda s}\mu_s\, \dd s\leq x\,\Big|\,\mu_0=\rho_1, \quad 
	\#(h) >1\right)\cdot \Pb(\#(h) >1\,|\,\mu_0=\rho_1).
	\end{align}
	Let I, II, III denote the first, the second, and the third term of the right-hand-side of the equation above, respectively. We proceed to obtain the following results
	\begin{align}
	\text{I}&= \Pb\left(\int_h^te^{-\lambda s}\mu_s\, \dd s\leq x-\int_0^he^{-\lambda s}\rho_1\, \dd s\,\Big|\,\mu_0=\rho_1, \quad \#(h) = 0\right)\cdot e^{-\alpha h}\\
	&= \Pb\left(\int_h^te^{-\lambda s}\mu_s\,\dd s\leq x-\int_0^he^{-\lambda s}\rho_1\, \dd s\,\Big|\,\mu_h=\rho_1\right)\cdot (1-\alpha h+o(h))\\
	&= \Pb\left(\int_h^te^{-\lambda (s-h)}\mu_s\,\dd s\leq e^{\lambda h} \left(x-\int_0^he^{-\lambda s}\rho_1\, \dd s \right)\,\Bigg|\,\mu_h=\rho_1\right)\cdot (1-\alpha h+o(h))\\
	&=\Pb \left(\int_h^t e^{-\lambda (s-h)}\mu_s\,\dd s \leq x-h(\rho_1-\lambda x)+o(h)\,\Big|\,\mu_h=\rho_1\right)\cdot (1-\alpha h+o(h))\\
	&=\Pb \left(\int_0^{t-h} e^{-\lambda s}\mu_s\,\dd s \leq x-h(\rho_1-\lambda x)+o(h)\,\Big|\,\mu_0=\rho_1\right)\cdot (1-\alpha h+o(h))\\
	&=u(t-h,x-h(\rho_1-\lambda x)+o(h)) \cdot (1-\alpha h+o(h))\\
	&=(1-\alpha h) \cdot u(t-h,x-h(\rho_1-\lambda x))+o(h),\\
	\text{II}&\leq \Pb\left(\int_h^te^{-\lambda s}\mu_s\,\dd s\leq x-\int_0^he^{-\lambda s}\rho_1\, \dd s\,\Big|\,\mu_0=\rho_1, \quad \#(h) = 1\right)\cdot (\alpha h+o(h))\\
	&=\Pb\left(\int_h^te^{-\lambda s}\mu_s\,\dd s\leq x-\int_0^he^{-\lambda s} \rho_1 \,\dd s\,\Big|\,\mu_h=\rho_2\right)\cdot (\alpha h+o(h))\\
	&=\Pb\left(\int_h^te^{-\lambda (s-h)}\mu_s\,\dd s\leq e^{\lambda h}\left(x-\int_0^he^{-\lambda s}\rho_1\, \dd s\right)\,\Bigg|\,\mu_h=\rho_2\right)\cdot (\alpha h+o(h))\\
	&=v(t-h,x+ O(h))\cdot (\alpha h+o(h))\\
	&=\alpha h \cdot v(t,x)+o(h).
	\end{align}
	Similarly, we can show that
	\begin{align}
	\text{II}&\geq \Pb \left(\int_h^te^{-\lambda s}\mu_s\,\dd s\leq x-\int_0^he^{-\lambda s}\rho_2\,\dd s\,\Big|\,\mu_0=\rho_1,\quad \#(h)=1\right)\cdot (\alpha h+o(h))\\
	&=\alpha h \cdot v(t,x)+o(h).
	\end{align}
	Obviously, $\text{III} = o(h)$.
	
	Therefore,
	\begin{align}
	u(t,x)=(1-\alpha h) \cdot u(t-h,x-h(\rho_1-\lambda x))+ \alpha h \cdot v(t,x)+o(h),
	\end{align}
	which implies the PDE of $u$ in \eqref{eqn_MC_ODE_system}. Similarly we can show the PDE of $v$ in \eqref{eqn_MC_ODE_system} holds as well. Moreover,
	\begin{align}
	u\left(t,\frac{\rho_1}{\lambda}(1-e^{-\lambda t})\right)&= \Pb \left(\int_0^te^{-\lambda s}\mu_s\,\dd s\leq \frac{\rho_1}{\lambda}(1-e^{-\lambda t})\,\Big|\,\mu_0=\rho_1\right)\\
	&= \Pb(\mu_s = \rho_1, \; \forall \, s \in (0,t]\,|\,\mu_0=\rho_1) \\
	&=e^{-\alpha t}.
	\end{align}
	The other boundary conditions are obvious.
\end{proof}

\subsection{Proof of Lemma \ref{lemma_limit_mu0}}
\label{appen_lemma_limit_mu0}
\begin{proof}
Recall that $u_\infty$ and $v_\infty$ are the conditional c.d.f. of $\int_0^\infty e^{-\lambda s}\mu_s\, \dd s$ given $\mu_0=\rho_1$ and $\mu_0=\rho_2$ respectively, see \eqref{eqn_def_u_v_infty}.

	Then the conditional p.d.f. of $(Q_\infty \, | \, \mu_0=\rho_1)$ is given by
	\begin{align}
	p_\infty(x):=\frac{\dd}{\dd x}
	\Pb\left( Q_\infty \le x | \mu_0 = \rho_1\right) =
	\int_{-\infty}^\infty u_\infty(z) \cdot \phi_\infty'(x-z)\, \dd z,
	\end{align}
	and the conditional p.d.f. of $(Q_\infty \, | \, \mu_0=\rho_2)$ is given by
	\begin{align}
	q_\infty(x):=
	\frac{\dd}{\dd x}
	\Pb\left( Q_\infty \le x | \mu_0 = \rho_2\right) =
	\int_{-\infty}^\infty v_\infty(z)\cdot \phi_\infty'(x-z)\, \dd z,
	\end{align}
	where $\phi_\infty$ is the p.d.f. of $\sigma\int_0^\infty e^{-\lambda s}\,dW_s-\frac{\sigma^2}{2\lambda}$, i.e.,
	\begin{equation}\label{eq12}
	\phi_\infty(x) =\sqrt{\dfrac{\lambda}{\pi\sigma^2}}\cdot \exp \left( -\dfrac{\left[x+\frac{\sigma^2}{2\lambda}\right]^2}{\frac{\sigma^2}{\lambda}} \right).
	\end{equation}
	Similar to $\Eb[\mu_0 | Q_t]$ in \eqref{eq10}, we obtain
	\begin{equation}\label{eq11}
	\Eb[\mu_0|Q_\infty]=\frac{\rho_1\beta \cdot p_\infty(Q_\infty)+\rho_2\alpha \cdot q_\infty(Q_\infty)}{\beta \cdot p_\infty(Q_\infty)+\alpha \cdot q_\infty(Q_\infty)}.
	\end{equation}
	Equations \eqref{eq10} and \eqref{eq11} indicate that, in order to show \eqref{eq8}, it suffices to show that $p(t, \cdot)$ and $q(t, \cdot)$ uniformly converge to $p_\infty(\cdot)$ and $q_\infty(\cdot)$, respectively. We will use \cite[Lemma 1]{MR773179} to prove such a result.
	
	We have that
	\begin{align}
	p(t, x)=\int_{-\infty}^\infty u(t,z)\cdot\sqrt{\frac{\lambda}{\pi\sigma^2\left(1-e^{-2\lambda t}\right)}}\cdot \frac{-2\left[x-z+\frac{\sigma^2}{2\lambda}\left(1-e^{-\lambda t}\right)\right]}{\frac{\sigma^2}{\lambda}\left(1-e^{-2\lambda t}\right)}\cdot e^{-\dfrac{\left[x-z+\frac{\sigma^2}{2\lambda}\left(1-e^{-\lambda t}\right)\right]^2}{\frac{\sigma^2}{\lambda}\left(1-e^{-2\lambda t}\right)}}\, \dd z.
	\end{align}
	For $t>1$, these exists some constant $K>0$ such that
	\begin{align}
	p(t, x)&\leq K\int_{-\infty}^\infty \left|x-z+\frac{\sigma^2}{2\lambda}\left(1-e^{-\lambda t}\right)\right|\cdot e^{-\dfrac{\left[x-z+\frac{\sigma^2}{2\lambda}\left(1-e^{-\lambda t}\right)\right]^2}{\frac{\sigma^2}{\lambda}\left(1-e^{-2\lambda t}\right)}}\, \dd z\\
	&= K \int_{-\infty}^\infty |z|\cdot e^{-\dfrac{z^2}{\frac{\sigma^2}{\lambda}\left(1-e^{-2\lambda t}\right)}}\, \dd z\leq K\int_{-\infty}^\infty |z|\cdot e^{-\frac{\lambda z^2}{\sigma^2}}\, \dd z<\infty.
	\end{align}
	Denote 
	\begin{align}
	\theta_t:=\frac{\sigma^2}{2\lambda}\left(1-e^{-\lambda t}\right) \quad \text{ and } \quad 
	\zeta_t:=\frac{\sigma^2}{\lambda}\left(1-e^{-2\lambda t}\right).
	\end{align}
	Let $t>\frac{1}{\lambda}$ and $0 < |x-y|\leq \eps$ for some positive constant $\eps$, we have that
	\begin{align}
	|p(t,x)-p(t,y)|&\leq K \int_{-\infty}^\infty\left|\left(x-z+\theta_t\right)\cdot e^{-\dfrac{\left(x-z+\theta_t\right)^2}{\zeta_t}}-\left(y-z+\theta_t\right)\cdot e^{-\dfrac{\left(y-z+\theta_t\right)^2}{\zeta_t}}\right|\, \dd z\\
	&=K \int_{-\infty}^\infty \left|((x-y)+z)\cdot e^{-\dfrac{((x-y)+z)^2}{\zeta_t}}-z\cdot e^{-\dfrac{z^2}{\zeta_t}}\right|\, \dd z\\
	&\leq K \int_{-\infty}^\infty |x-y|\cdot e^{-\dfrac{((x-y)+z)^2}{\zeta_t}}\, \dd z+ K\int_{-\infty}^\infty|z|\cdot \left|e^{-\dfrac{((x-y)+z)^2}{\zeta_t}}-e^{-\dfrac{z^2}{\zeta_t}}\right|\, \dd z\\
	&\leq K \eps \int_{-\infty}^\infty e^{-\dfrac{((x-y)+z)^2}{\zeta_t}}\,\dd z+K\int_{-\infty}^\infty|z|\cdot e^{-\dfrac{z^2}{\zeta_t}}\cdot \left|e^{\dfrac{z^2}{\zeta_t}-\dfrac{((x-y)+z)^2}{\zeta_t}}-1\right|\, \dd z\\
	&\leq K \eps\int_{-\infty}^\infty e^{-\dfrac{z^2}{\zeta_t}}\,\dd z+K \int_{-\infty}^\infty|z|\cdot e^{-\dfrac{\lambda z^2}{\sigma^2}}\cdot \left(e^{\left|\dfrac{z^2}{\zeta_t}-\dfrac{((x-y)+z)^2}{\zeta_t}\right|}-1\right)\,\dd z\\
	&\leq K \eps\int_{-\infty}^\infty e^{-\dfrac{\lambda z^2}{\sigma^2}}\,\dd z+K\int_{-\infty}^\infty|z|\cdot e^{-\dfrac{\lambda z^2}{\sigma^2}}\cdot \left(e^{\dfrac{2\lambda}{\sigma^2}\eps(\eps+2|z|)}-1\right)\, \dd z\\
	&\rightarrow 0,\quad \eps \rightarrow 0,
	\end{align}
	where $K$ is some constant that is independent of $t, x, y, \eps$ and may vary from line to line, and the limit result on the last line follows from the dominated convergence theorem.
	
	Moreover,
	\begin{align}
	\lim_{x\rightarrow\pm\infty}p_\infty(x)&=\lim_{x\rightarrow\pm\infty}\int_{-\infty}^\infty u_\infty(z)\cdot\sqrt{\dfrac{\lambda}{\pi\sigma^2}}\cdot \dfrac{-2\left(x-z+\frac{\sigma^2}{2\lambda}\right)}{\frac{\sigma^2}{\lambda}}\cdot e^{-\dfrac{\left(x-z+\frac{\sigma^2}{2\lambda}\right)^2}{\frac{\sigma^2}{\lambda}}}\, \dd z\\
	&=-\dfrac{2\lambda}{\sigma^2}\sqrt{\dfrac{\lambda}{\pi\sigma^2}}\int_{-\infty}^\infty \lim_{x\rightarrow\pm\infty}u_\infty\left(x+\frac{\sigma^2}{2\lambda}-z\right)\cdot z\cdot e^{\dfrac{\lambda z^2}{\sigma^2}}\, \dd z\\
	&=0,
	\end{align}
	where the third equality follows from the dominated convergence theorem. 
	
	Thus, by \cite[Lemma 1]{MR773179}, $p(t,\cdot)$ uniformly converges to $p_\infty(\cdot)$. Similarly, we can show that $q(t, \cdot)$ uniformly converges to $q_\infty(\cdot)$.
\end{proof}

\subsection{Proof of Lemma \ref{lem_limit_distribution}}
\label{appen_lem_limit_distribution}
\begin{proof}
	Similar to the argument in the proof of Proposition \ref{p1}, we can show that $u_\infty$ and $v_\infty$ satisfy the following ordinary differential equations (ODE) system:
	\begin{empheq}[left={\empheqlbrace}]{align}
	& (\rho_1-\lambda x)u_\infty'+\alpha u_\infty -\alpha v_\infty=0,\quad\frac{\rho_1}{\lambda}<x<\frac{\rho_2}{\lambda}\\
	& (\rho_2-\lambda x)v_\infty'+\beta v_\infty -\beta u_\infty=0,\quad\frac{\rho_1}{\lambda}<x<\frac{\rho_2}{\lambda},\\
	& u_\infty\left(\frac{\rho_1}{\lambda}\right)=0,\quad u_\infty\left(\frac{\rho_2}{\lambda}\right)=1,\\
	& v_\infty\left(\frac{\rho_1}{\lambda}\right)=0,\quad v_\infty\left(\frac{\rho_2}{\lambda}\right)=1.
	\end{empheq}
	Notice that, from the second ODE, we have 
	\begin{align}
	u_\infty = v_\infty + \frac{\rho_2 - \lambda x}{\beta}v_\infty', \quad \text{ and then } \quad 
	u_\infty' = \(1-\frac{\lambda}{\beta} \) v_\infty' + \frac{\rho_2 - \lambda x}{\beta} v_\infty''.
	\end{align}
	Next, by plugging the above expressions into the first ODE, we obtain a second-order ODE that involves $v_\infty$ only. Solving such an ODE of $v_\infty$ yields $v_\infty$, which further implies $u_\infty$ is as stated in the lemma.
\end{proof}

\bibliographystyle{apalike}
\bibliography{references}

\end{document}